\documentclass[a4paper,11pt,oneside]{article}
\usepackage[a4paper, margin=1in]{geometry}
\usepackage{setspace}
\usepackage{textcomp}
\usepackage{enumerate}
\usepackage{amsmath,amssymb,amsfonts,amsthm,mathrsfs}
\usepackage{courier}
\usepackage{algpseudocode,algorithm}
\usepackage{cancel, cases}
\usepackage{soul}
\usepackage{float,graphicx,subcaption,epstopdf}
\usepackage{longtable,tabu,threeparttable,rotating,adjustbox,booktabs}
\usepackage[normalem]{ulem}
\usepackage[dvipsnames]{xcolor}
\usepackage{tcolorbox} 

\usepackage{hyperref}

\usepackage{caption}

\usepackage[round, authoryear]{natbib}

\DeclareFontFamily{OT1}{pzc}{} \DeclareFontShape{OT1}{pzc}{m}{it}%
{<-> s * [0.900] pzcmi7t}{} \DeclareMathAlphabet{\mathpzc}{OT1}{pzc}%
{m}{it}

\newcommand\tr{\text{tr}}

\newcommand{\bi}{\begin{itemize}}
\newcommand{\ei}{\end{itemize}}
\DeclareMathOperator*{\argmin}{argmin}
\DeclareMathOperator*{\argmax}{argmax}

\theoremstyle{plain}
\newtheorem{theorem}{Theorem}
\newtheorem{corollary}[theorem]{Corollary}

\newtheorem{proposition}[theorem]{Proposition}

\theoremstyle{definition}
\newtheorem{definition}[theorem]{Definition}
\newtheorem{assumption}[theorem]{Assumption}

\theoremstyle{remark}
\newtheorem{remark}{Remark}

\numberwithin{equation}{section}
\numberwithin{theorem}{section}

\title{Reinforcement Learning for Risk-Sensitive Investment Management: a Free Energy--Entropy Duality Approach}
\author{S\'ebastien Lleo\footnote{Finance Department and `AI, Data Science \& Business' AE, NEOMA Business School, France} \and Wolfgang Runggaldier\footnote{University of Padova, Italy, and Fellow, Institut Louis Bachelier, Paris, France}}
\date{\today}

\begin{document}
\maketitle

\begin{abstract}
This paper develops a reinforcement-learning approach to continuous-time risk-sensitive benchmarked asset allocation in a partly model-based setting. The benchmarked problem does not directly fit the standard Markovian stochastic-control template: the state is uncontrolled, whereas the terminal reward contains a controlled It\^o integral. We use free energy--entropy duality to reformulate the problem as a linear--quadratic--Gaussian stochastic differential game under an equivalent probability measure, yielding explicit finite- and infinite-horizon saddle-point solutions. This structure guides a continuous-time $q$-learning actor--critic method: the quadratic value function motivates the critic, while the affine saddle-point controls motivate deterministic actors for the portfolio allocation and adversarial control. The learned allocation admits an economic interpretation through fractional Kelly decompositions. A proof-of-concept implementation calibrated to U.S. equity data shows that the actors learn the optimal policy with high accuracy and reveals a favorable asymmetry: the portfolio actor receives a cleaner learning signal than the auxiliary adversarial actor.
\end{abstract}

\textbf{Keywords:} actor--critic methods, Fractional Kelly strategies, Free energy--entropy duality, Policy gradient, Reinforcement learning, Risk-sensitive control, Stochastic games.

\textbf{JEL Classification}: C32; C44; C61; C73; G11; G12.

\textbf{MSC classes:} 68T05, 91G10; 91G80; 93E20.

\section{Introduction}\label{sec:Intro}

In continuous-time asset allocation models, investors are typically assumed to know the model parameters or to estimate them with sufficient accuracy. In practice, however, expected returns, factor exposures, and covariance matrices are difficult to estimate. Moreover, estimation error can have a substantial effect on optimal portfolios, even in classical mean--variance settings \citep{Chopra1993}. One response is to model the problem under partial observation, leading to a joint estimation--control problem and raising questions about the separation of estimation and control \citep[see][]{LleoRunggaldierSeparation24}. Reinforcement learning (RL) offers a complementary route \citep{SuttonBarto2018}: it seeks to learn investment policies from data and feedback, thereby reducing reliance on fully specified and precisely estimated models.

This paper develops an RL approach to a continuous-time risk-sensitive benchmarked asset allocation problem in a partly model-based setting. The investment problem is motivated by a price-taking asset manager who observes market prices, allocates across a factor-driven investment universe, and is evaluated relative to a stochastic benchmark. The RL formulation uses the structure of the analytical solution to design interpretable actor and critic parametrizations, while relaxing the assumption that all model coefficients are known.

The starting point is the risk-sensitive benchmarked asset allocation model. This model is a natural foundation for financial engineering applications: the market is factor-based, the performance criterion is measured against a benchmark, the risk-sensitive objective is connected to both utility theory and continuous-time mean--variance optimization, and the optimal allocation admits an interpretation through the Kelly criterion. At the same time, the benchmarked problem does not directly fit the standard Markovian stochastic-control template. The state process is uncontrolled, while the terminal reward contains a controlled It\^o integral that does not disappear under expectation. We therefore use the free energy--entropy duality of \citet{LleoRunggaldier_RSIMviaDuality_2026} to reformulate the benchmarked risk-sensitive problem as a linear--quadratic--Gaussian stochastic differential game under an equivalent probability measure. The resulting game has a quadratic value function and explicit affine saddle-point controls.

The paper first develops the finite- and infinite-horizon versions of this duality-based formulation. The finite-horizon case provides the foundation for episodic actor--critic learning, while the infinite-horizon case leads to an ergodic game and a continuing-task RL formulation. The ergodic setting is especially relevant for institutional investment management, where the investment organization is naturally modeled as a going concern focused on long-run risk-adjusted performance.

Building on this stochastic-game representation, we develop continuous-time $q$-learning actor--critic methods for the benchmarked investment problem. Our construction follows the martingale approach to continuous-time RL developed by \citet{jiaPolicyEvaluationTemporalDifference2022,jiaPolicyGradientActor2022,jiaQLearningContinuousTime}, and is related to the stochastic-control perspective on continuous-time RL in \citet{wangReinforcementLearningContinuous2020}. We specialize this framework to the LQG game generated by the duality. This specialization is useful: the quadratic form of the value function motivates the critic parametrization, while the affine form of the saddle-point controls motivates the actor parametrization. The target actors are deterministic Markov policies in the sense of \citet{silverDeterministicPolicyGradient2014}: they map the time--state pair directly into the portfolio control and the adversarial control. Exploration, when used for learning, is introduced through behavior policies rather than through the target policy itself.

The paper contributes to the literature in four ways. First, it develops a reinforcement-learning formulation and implementation for a continuous-time risk-sensitive investment problem. The closest related work in continuous-time risk-sensitive RL is \citet{jiaContinuousTimeRiskSensitiveReinforcement2026}, who studies continuous-time risk-sensitive RL under an entropy-regularized exploratory diffusion formulation. Our approach is complementary. Jia considers a general risk-sensitive RL formulation with exploratory policies, whereas we start from a benchmarked investment problem, use free energy--entropy duality to recast it as an LQG stochastic differential game, and then design deterministic actor--critic updates adapted to the resulting structure. Both approaches highlight the importance of learning with both a state-value function and a $q$-function.

Second, the paper contributes to the emerging literature on continuous-time reinforcement learning with deterministic policies. In discrete-time RL with continuous action spaces, deterministic policy gradient methods are not merely a limiting case of stochastic policy gradient methods \citep{silverDeterministicPolicyGradient2014}. An analogous distinction arises in continuous time. \citet{chengDeterministicPolicyGradient2026} derive a continuous-time deterministic policy gradient formula based on an advantage-rate function, establish a martingale characterization, and use it to motivate a model-free continuous-time deterministic policy gradient algorithm. Our paper takes a different but related route: rather than pursuing a fully model-free method, we use the analytical structure of a financial stochastic game to reduce model dependence while retaining interpretability.

Third, the paper contributes to the growing literature on RL for portfolio selection and asset management, including \citet{kolmModernPerspectivesReinforcement2020,kolmReinforcementLearningAsset2025,halperinChapter6Reinforcement2025}, \citet{jiangReinforcementLearningKelly2022}, and \citet{lleo2026exploratoryrandomizationRSBAM}. Much of this literature formulates portfolio RL through risk-adjusted rewards, transaction costs, or other dynamic portfolio objectives. By contrast, our model starts from continuous-time risk-sensitive control and benchmarked performance. It is therefore inherently dynamic, factor-based, and benchmark-relative. The Kelly portfolio appears as a limiting case, while the benchmarked allocation admits a fractional Kelly interpretation.

Fourth, the paper addresses interpretability. A common criticism of RL methods in investment management is that their recommendations may be difficult to explain to clients, investment committees, and regulators. We show how the structure of the benchmarked risk-sensitive solution can be used to decompose the learned allocation into a Kelly portfolio, a benchmark-tracking portfolio, and an intertemporal hedging portfolio. This decomposition makes the RL output economically interpretable at a limited additional computational cost.

Finally, we provide a proof-of-concept implementation for the continuing case. The implementation considers a six-factor U.S. equity allocation problem with $m=13$ exchange-traded funds and a replicable benchmark. We use these data to compute the parameters for the analytical ergodic risk-sensitive benchmarked asset management model. The learning speed and accuracy of the actor--critic algorithm are assessed against this analytical model. The experiment focuses on learning the optimal portfolio and adversarial policies, while the value function is computed from the analytical model. The algorithm recovers the analytical saddle point with high accuracy. The diagnostics also reveal a structural asymmetry induced by the duality: the portfolio-gradient critic is learned essentially exactly, while the adversarial-gradient critic retains a small residual error because the adversarial control enters the state transition. This asymmetry is favorable in asset-management applications, since the portfolio allocation is the economically relevant control.

The remainder of the paper is organized as follows. Section~\ref{sec:RSBAM} presents the finite-horizon duality-based formulation. Section~\ref{sec:RSBAM:infinite} develops the ergodic counterpart. Section~\ref{sec:PolicyGradientMethods} derives the continuous-time actor--critic framework and discusses explainability through fractional Kelly decompositions. Section~\ref{sec:Implementation} presents the proof-of-concept implementation. Section~\ref{sec:Conclusion} concludes.

\section{Risk-Sensitive Benchmarked Asset Management Problem: Solution via the Free Energy--Entropy Duality}
\label{sec:RSBAM}

This section sets out the finite-horizon benchmarked risk-sensitive investment problem and its dual stochastic-game representation. The formulation follows \citet{LleoRunggaldier_RSIMviaDuality_2026} and provides the finite-horizon duality foundation for the reinforcement-learning methods developed below.
The setting is a continuous-time factor model with benchmarked performance, exponential risk sensitivity, and an equivalent stochastic differential game obtained through the free energy--entropy duality.

Let $\left(\Omega,\mathcal F,(\mathcal F_s)_{0\leq s\leq T},\mathbb P\right)$
be a filtered complete probability space supporting a $d$-dimensional $(\mathcal F_s)$-Wiener process $W$, where $d:=n+m+1$ for integers $n,m\geq1$. We use terminology that connects stochastic control and reinforcement learning. A \emph{strategy} or \emph{policy} is a stochastic process $H=(h_s)_{s\in[t,T]}$. The value $h_s\in\mathbb R^m$ is the \emph{control}, or \emph{action}, at time $s$. A Markov \emph{strategy} or \emph{policy} is a measurable map
$
h:[t,T]\times\mathbb R^n\to\mathbb R^m,
$
so that the Markov \emph{control} $h_s = h(s,X_s)$.

\subsection{Model for the Financial Market}
\label{sec:RSBAM:model}

The investor trades in $m$ risky financial assets $S = (S_s)_{s\in[0,T]}$. Their discounted prices satisfy
\begin{align}\label{eq:dS}
dS_s
=
\text{diag}(S_s)\left[(a_s+A_sX_s)\,ds+\Sigma_s,dW_s\right],
\end{align}
where $\text{diag}(\cdot)$ maps a vector into a diagonal matrix. The factor process $X=(X_s){s\in[0,T]}$, also called the state process, represents financial and economic factors affecting asset returns and follows
\begin{align}\label{eq:state}
    dX_s=(b_s+B_sX_s)\,ds+\Lambda_s,dW_s.
\end{align}
The initial state $X_t$ is either deterministic, with $X_t=x$, or random. In the latter case, $X_t\sim N(\mu,P)$, with known $\mu\in\mathbb R^n$ and $P\in\mathbb R^{n\times n}$, independently of $W$. The discounted benchmark index $L=(L_s)_{s\in[0,T]}$ satisfies
\begin{align}
\frac{dL_s}{L_s}
=
(c_s+C_sX_s)\,ds+\Xi_s'dW_s.
\label{eq:dL}
\end{align}

Throughout the paper, we write $a_s, A_s,\Sigma_s$, etc., for $a(s),A(s),\Sigma(s)$, etc. Vectors are column vectors unless otherwise stated, and a prime denotes matrix transposition. For notational convenience, $C_s$ is treated as a row vector in expressions involving the benchmark drift, so that we write $C_sX_s$ rather than $C_s'X_s$. Moreover, the coefficients of equations \eqref{eq:dS}-\eqref{eq:dL} satisfy the following standing assumption:
\begin{assumption}\label{as:coefficients}
    The coefficient functions in \eqref{eq:dS}--\eqref{eq:dL} are continuously differentiable on $[0,T]$ and have the dimensions specified above.
\end{assumption}

We also assume that the risky assets have a nonsingular instantaneous covariance matrix:
\begin{assumption}\label{as:sigma:posdef}
For every \(s\in[0,T]\), the matrix \(\Sigma_s\Sigma_s'\) is positive definite.
\end{assumption}

We model the self-financing investment policy as an $\mathcal{F}_s$-adapted process $H = \left(h_s\right)_{s\in [0,T]} \in \mathbb{R}^m$ belonging to the class $\mathcal{A}^H_T$.

\begin{definition}[Admissible (Target) Investment Policies]\label{def:classAT:fullobs}

A \emph{policy} $H = \left(h_s\right)_{s \in [0,T]} \in \mathbb{R}^{m}$ is in class $\mathcal{A}^{H}_T$ if the \emph{control} process $h_s$ is a Markov control and $\mathbb P\left(\int_{0}^{T} \left| h_s \right|^2 \, ds < +\infty \right) = 1$.

\end{definition}

The \emph{discounted} wealth process $\left(V_s\right)_{s \in [0,T]}$ generated by the investment policy $H$ solves the SDE:
\begin{align}\label{eq:V}
\frac{dV_s}{V_s}
&= h_s'\left(a_s+ A_s X_s\right) \, ds
+ h_s' \Sigma_s  dW_s.
\end{align}
To track the investment portfolio's cumulative outperformance relative to its benchmark, we introduce the \emph{log price relative} process $\left(R_s\right)_{s \in [0,T]}$, defined as $R_s:= \ln \frac{V_s}{L_s}$. This process solves the SDE:
\begin{align}\label{eq:excess_return}
dR_s = \left[ 
  \left(- \frac{1}{2} h_s'\Sigma_s\Sigma_s'h_s 
  + h_s'a_s 
  + \frac{1}{2}\Xi_s'\Xi_s - c_s \right)  
  + \left(h_s' A_s - C_s \right)X_s\right] \, ds
  + \left(h_s'\Sigma_s - \Xi_s' \right) \, dW_s. 
\end{align}
From here, we obtain the \emph{log excess return} of the portfolio over its benchmark as the difference $R_T - R_0$.

The investor seeks to maximize the risk-sensitive benchmarked criterion: 
\begin{align}\label{eq:J} 
J_T(H,\theta)
&:= -\frac{1}{\theta} \ln \mathbf{E} \left[ e^{-\theta (R_T-R_0)} \right]
,
\end{align}
where $\theta \in (-1,0)\cup(0,\infty)$ is the risk-sensitivity parameter and $T<\infty$ is a fixed time horizon. Equivalently, one may minimize the exponentially transformed criterion
\begin{align}\label{eq:I}
I_T(H,\theta)
&:=
e^{-\theta J_T(H,\theta)}
=
\mathbf{E}\left[e^{-\theta(R_T-R_0)}\right].
\end{align}
The case $\theta>0$ is the primary one, as noted in \citet{LleoRunggaldier_RSIMviaDuality_2026}, so we focus on it in this paper.

\subsection{Using the Free Energy--Entropy Duality to Express the Risk-Sensitive Control Problem as a Stochastic Differential Game}\label{sec:RSBAM:RSC_problem}

Using the \emph{free energy--entropy duality}, Section 2.2 in \citet{LleoRunggaldier_RSIMviaDuality_2026} proved that 
\begin{align}\label{eq:EEDuality:inf}
    \inf_{H \in \mathcal{A}^H_T} I_T(H,\theta) 
    = \exp\left\{ \inf_{H \in \mathcal{A}^H_T} \sup_{ \Gamma \in \mathcal{A}^\Gamma_T} \mathbf{E}^{\mathbb{P}^\Gamma} \left[ \theta \int_{0}^{T} g(s,X_s,h_s,\gamma_s;\theta) \, ds
    \right] 
    \right\}
,
\end{align}
where $\Gamma = \left(\gamma_s\right)_{s \in [0,T]} \in \mathbb{R}^d$ is an $\mathcal{F}_{s}$-adapted, a.s. square integrable process belonging to the class $\mathcal{A}^\Gamma_T$ which we define below, $\mathbf{E}^{\mathbb{P}^\Gamma} \left[ \cdot \right]$ denotes the expectation taken with respect to an equivalent measure $\mathbb{P}^\Gamma$, and where
\begin{align}\label{eq:g}
    g(s,x,h,\gamma;\theta) 
    :=  \frac{1}{2} h'\Sigma_s\Sigma_s'h
        - h' a_s 
        - \frac{1}{2}\Xi_s'\Xi_s 
        + c_s  
        - \left(h'\Sigma_s - \Xi_s' \right)\gamma
        - \left(h' A_s - C_s \right)x 
        - \frac{1}{2\theta} \|\gamma\|^2.
\end{align}

The measure $\mathbb{P}^\Gamma$ is parameterized by the process $\Gamma$. By Girsanov's theorem, the Radon--Nikodym derivative $\frac{d\mathbb{P}^\Gamma}{d\mathbb{P}}$ restricted to $\mathcal{F}_T$ is 
\begin{align}\label{eq:RNderivative:gamma}
    \frac{d\mathbb{P}^\Gamma}{d\mathbb{P}}\Big{|}_{\mathcal{F}_T}
 = \exp \left\{ -\frac{1}{2} \int_0^T  \| \gamma_s \|^2 ds + \int_0^T \gamma_s' dW_s \right\} =: \chi^\Gamma_{[0,T]}, 
\end{align} 
and the $\mathbb{P}^\Gamma$-Wiener process $W^\Gamma$ is given by
$W^\Gamma_s := W_s - \int_0^s \gamma_s \, ds$.
 
Under the measure $\mathbb{P}^\Gamma$, the dynamics of the state process $X$ satisfies
\begin{align}\label{eq:state:Pgamma:FO}
    d X_s
    = \left(b_s + B_s X_s + \Lambda_s \gamma_s \right) \, ds
    + \Lambda_s \, dW^\Gamma_s,
\end{align}
and analogously for \eqref{eq:dS} and \eqref{eq:dL}. Therefore, the SDE for the log price relative process $R$ under $\mathbb{P}^\Gamma$ is
\begin{align}\label{eq:excess_return:Pgamma:FO}
&   dR_s 
                            \nonumber\\
=& \left[ 
    - \frac{1}{2} h_s'\Sigma_s\Sigma_s'h_s 
  + h_s'a_s 
  + \frac{1}{2}\Xi_s'\Xi_s - c_s  
  + \left(h_s'\Sigma_s - \Xi_s' \right)\gamma_s
  + \left(h_s' A_s - C_s \right)X_s\right] \, ds
  + \left(h_s'\Sigma_s - \Xi_s' \right) \, dW^\Gamma_s.
\end{align}

Finally, the class $\mathcal{A}^\Gamma_T$ of policies $\Gamma$ is:

\begin{definition}[Admissible (Target) Adversarial Policies]\label{def:classAgammaT:fullobs}

A \emph{policy} $\Gamma = \left(\gamma_s\right)_{s \in [0,T]} \in \mathbb{R}^{d}$  is in class $\mathcal{A}^\Gamma_T$ if (i) the \emph{control} process $\gamma_s$ is a Markov control; (ii) $\mathbb P\left(\int_{0}^{T} \left| \gamma_s \right|^2 \, ds < +\infty \right) = 1$; and (iii) $\chi^\Gamma_{[0,T]}$, defined at \eqref{eq:RNderivative:gamma}, is an exponential martingale.

\end{definition}

The right-hand side of \eqref{eq:EEDuality:inf} defines a two-player stochastic differential game with objective
\begin{align}\label{eq:def:SDG}
\inf_{H \in \mathcal{A}^H_T} \sup_{ \Gamma \in \mathcal{A}^\Gamma_T} \mathbf{E}^{\mathbb{P}^\Gamma} \left[ \theta \int_{0}^{T} g(s,X_s,h_s,\gamma_s;\theta) \, ds
\right],
\end{align}
where the strategy of the minimizing player is $H = \left(h_s\right)_{s \in [0,T]} \in \mathbb{R}^{m}$ in class $\mathcal{A}^{H}_T$, the strategy of the maximizing player is $\Gamma = \left(\gamma_s\right)_{s \in [0,T]} \in \mathbb{R}^{d}$ in class $\mathcal{A}^\Gamma_T$, and the state process dynamics is given at \eqref{eq:state:Pgamma:FO}.

\subsection{Finite-Horizon Duality-Based Solution}\label{sec:mainresult}

For a generic $t\in[0,T)$  and with $R_t=0$, the criteria in \eqref{eq:J} and \eqref{eq:I} become
\begin{align}\label{eq:IJ:general}
    J_T(t,x;H,\theta) := -\frac{1}{\theta} \ln \mathbf{E}_{t,x} \left[ e^{-\theta R_T} \right],
    \qquad
    I_T(t,x;H,\theta) 
    := e^{-\theta J_T(t,x;H,\theta)} 
    = \mathbf{E}_{t,x} \left[ e^{-\theta R_T} \right],
\end{align}
so the duality at \eqref{eq:EEDuality:inf} generalizes to
\begin{align}\label{eq:EEDuality:inf:SDG} 
\inf_{H \in \mathcal{A}^H_T} I_T(t,x;H,\theta) 
    = \exp\left\{ \inf_{H \in \mathcal{A}^H_T} \sup_{ \Gamma \in \mathcal{A}^\Gamma_T} \mathbf{E}_{t,x}^{\mathbb{P}^\Gamma} \left[ \theta \int_{t}^{T} g(s,X_s,h_s,\gamma_s;\theta) \, ds
    \right]\right\},
\end{align}
and the stochastic differential game at \eqref{eq:def:SDG} becomes
\begin{align}\label{eq:valuefunction:u}
    u(t,x) 
    = \inf_{H \in \mathcal{A}^H_T} \sup_{ \Gamma \in \mathcal{A}^\Gamma_T} \mathbf{E}_{t,x}^{\mathbb{P}^\Gamma} \left[ \theta \int_{t}^{T} g(s,X_s,h_s,\gamma_s;\theta) \, ds
    \right], 
\end{align}
where $u(t,x)$ is the optimal value function associated with the game\footnote{Solving the stochastic differential game is equivalent to solving the original control problem. Specifically, Lemma 1 in \citet{LleoRunggaldier_RSIMviaDuality_2026} shows that $u(t,x) = - \theta \sup_{H \in \mathcal{A}^H_T} J_T(t,x;H,\theta).$}. 

The stochastic differential game for $u(t,x)$ is a Linear-Quadratic-Gaussian (LQG) game with associated Bellman--Isaacs partial differential equation (PDE):
\begin{align}\label{eq:BellmanIsaacs:u}
    \frac{\partial u(s,x)}{\partial s}
    + \mathcal{H}\left(s,x,Du(s,x), D^2 u(s,x)\right)
    = 0
\end{align}
where $Du'(s,x) = \left( \frac{\partial u(s,y)}{\partial y_1}, \ldots, \frac{\partial u(s,y)}{\partial y_i}, \ldots, \frac{\partial u(s,y)}{\partial y_n}\right)\Big{|}_{y=x}$, $D^2u(s,x) = \left[ \frac{\partial^2 u(s,y)}{\partial y_i \partial y_j} \right] \Big{|}_{y=x}$ for $i,j=1,\ldots,n$, 
\begin{align}\label{eq:Hamiltonian:H+}
   \mathcal{H}(s,x, p, M) := \inf_{h \in \mathbb{R}^m} \sup_{\gamma \in \mathbb{R}^d} \left\{ \left[b_s + B_s x + \Lambda_s \gamma \right]' p 
    + \frac{1}{2} \tr \left(\Lambda_s\Lambda_s' M \right)
    +\theta g(s,x,h,\gamma;\theta) \right\}, 
\end{align}
with terminal condition $u(T,x) = 0$. \citet{LleoRunggaldier_RSIMviaDuality_2026} showed that the order in which the infimum and supremum are performed in the Hamiltonian does not affect the solution of the game.

\begin{theorem}[Finite-horizon duality-based solution]\label{theo:main_previouspaper}

The value function $u$ defined at \eqref{eq:valuefunction:u} is quadratic of the form 
    \begin{align}\label{eq:sol:u}
        u(t,x) = -\theta \left( \frac{1}{2} x' Q_t x + q_t'x + k_t \right),
    \end{align}
    where $Q: [t,T] \to \mathbb{R}^{n \times n}, q : [t,T] \to \mathbb{R}^n$, and $k: [t,T] \to \mathbb{R}$ are three functions such that:
    \begin{enumerate}[(i)]
    \item $Q_t$ is the unique symmetric non-negative definite solution to the matrix Riccati equation
    \begin{align}\label{eq:Q:Riccati}
        &\dot{Q}_s
            - \theta  Q_s'\Lambda_s \left( I_d- \frac{\theta}{\theta+1}\Sigma_s'(\Sigma_s\Sigma_s')^{-1} \Sigma_s\right) \Lambda_s'Q_s
            + \left(B_s' - \frac{\theta}{\theta+1} A_s'(\Sigma_s\Sigma_s')^{-1} \Sigma_s\Lambda_s' 
            \right)Q_s 
                                            \nonumber\\
        &+ Q_s'\left(B_s - \frac{\theta}{\theta+1} \Lambda_s\Sigma_s'(\Sigma_s\Sigma_s')^{-1}A_s \right)
            + \frac{1}{(\theta+1)} A_s' (\Sigma_s\Sigma_s')^{-1} A_s
        = 0, \quad Q_T=0;
    \end{align}
    \item $q_t$ solves the linear ODE
    \begin{align}\label{eq:q:ODE}
        & \dot{q}_s
        + \left(B_s'- \frac{\theta}{\theta+1} A_s'(\Sigma_s\Sigma_s')^{-1}\Sigma_s\Lambda_s' \right)q_s
        - \theta Q_s' \Lambda_s\left( I_d- \frac{\theta}{\theta+1} \Sigma_s'(\Sigma_s\Sigma_s')^{-1}\Sigma_s\right)\Lambda_s' q_s
                                        \nonumber\\
        &   + Q_s' \left[b_s - \frac{\theta}{\theta+1}  \Lambda_s\Sigma_s'(\Sigma_s\Sigma_s')^{-1}
    \left( a_s + \theta \Sigma_s\Xi_s \right) + \theta \Lambda_s\Xi_s\right]
                                        \nonumber\\
        &   - C_s'
        + \frac{1}{\theta+1} A_s'(\Sigma_s\Sigma_s')^{-1}
    \left( a_s + \theta \Sigma_s\Xi_s \right)
        = 0, \quad q_T =0;
    \end{align}
    \item $k_t$ is found by integration via
    \begin{align}\label{eq:k:integral}
        & k_t
        = \int_t^T \Bigg\{ 
        \frac{\theta}{2} q_s'\Lambda_s \left(I_d- \frac{\theta}{\theta+1} \Sigma_s'(\Sigma_s\Sigma_s')^{-1} \Sigma_s\right)\Lambda_s' q_s
                                        \nonumber\\
        &   - \left[b_s'+ \theta \Xi_s'\Lambda_s' - \frac{\theta}{\theta+1} \left( a_s + \theta \Sigma_s\Xi_s \right)'(\Sigma_s\Sigma_s')^{-1}\Sigma_s\Lambda_s' \right]q_s 
        - \frac{1}{2} \tr \left(\Lambda_s\Lambda_s' Q_s \right)
                                        \nonumber\\
        &   - \frac{1}{2(\theta+1)} \left(a_s + \theta \Sigma_s\Xi_s\right)'(\Sigma_s\Sigma_s')^{-1} \left(a_s + \theta \Sigma_s\Xi_s\right)
        + c_s
        + \frac{\theta-1}{2}\Xi_s'\Xi_s \Bigg\} \, ds, \quad k_T =0.
    \end{align}    
    \end{enumerate}  

Moreover, the pair of optimal strategies $(H^*, \Gamma^*)$ is defined by $H^* = \left(h^*(s,X_s) \right)_{s \in [t,T]} := \left(\hat{h}(s,X_s,Du(s,X_s)) \right)_{s \in [t,T]}$ and $\Gamma^* = \left(\gamma^*(s,X_s) \right)_{s \in [t,T]}  := \left(\hat{\gamma}(s,X_s,Du(s,X_s)) \right)_{s \in [t,T]}$, where the policies $\hat{h}(\cdot)$ and $\hat{\gamma}(\cdot)$ are given by
\begin{align}
        \hat{h}(s,x,Du(s,x))
            =&  \frac{1}{\theta+1}\left(\Sigma_s\Sigma_s'\right)^{-1}
        \left( a_s + A_s x
        + \theta \Sigma_s\Xi_s 
        + \Sigma_s\Lambda_s' Du(s,x)
        \right)
                            \label{eq:hhat}\\    
    \hat{\gamma}(s,x,Du(s,x)) 
=& \Lambda_s' Du(s,x) - \theta \left(\Sigma_s'\hat{h}(s,x,Du(s,x)) - \Xi_s \right)
                        \label{eq:gammahat:1}\\
    \phantom{\hat{\gamma}(s,x,Du(s,x))}
=&  \left[ I_d- \frac{\theta}{\theta+1} \Sigma_s'\left(\Sigma_s\Sigma_s'\right)^{-1}\Sigma_s \right]\Lambda_s' Du(s,x) 
        - \frac{\theta}{\theta+1} \Sigma_s'\left(\Sigma_s\Sigma_s'\right)^{-1}
        \left( a_s + A_s x \right) 
                                    \nonumber\\
        &+ \theta \left[ I_d- \frac{\theta}{\theta+1} \Sigma_s'\left(\Sigma_s\Sigma_s'\right)^{-1}\Sigma_s \right] \Xi_s,
                        \label{eq:gammahat:2}
    \end{align}   
or equivalently,
     \begin{align}
        \hat{h}(s,x,Du(s,x))
        =&  \left(\Sigma_s\Sigma_s'\right)^{-1}  
        (a_s + A_s x) + \left(\Sigma_s\Sigma_s'\right)^{-1}\Sigma_s \hat{\gamma}(s,x,Du(s,x))
                                                \label{eq:hhat:alt}\\
        \hat{\gamma}(s,x,Du(s,x)) 
=& \left[
            I_d+ \theta \Sigma_s' \left(\Sigma_s\Sigma_s'\right)^{-1}\Sigma_s
        \right]^{-1} 
        \left[ 
            - \theta \Sigma_s' \left(\Sigma_s\Sigma_s'\right)^{-1} (a_s + A_s x)
            + \theta \Xi_s
            + \Lambda_s'Du(s,x)
        \right]
                        \label{eq:gammahat:alt}
    \end{align}  

    Finally, the value function $u$ solves the stochastic differential game with the objective \eqref{eq:valuefunction:u}, in the sense that the saddle point condition
    \begin{align}\label{eq:verification:saddle}
            \mathbf{E}_{t,x}^{\mathbb{P}^\Gamma} \left[ \theta \int_{t}^{T} g(s,X_s,h^*_s,\gamma_s;\theta) \, ds \right]
            \leq 
            u(t,x) 
            \leq 
            \mathbf{E}_{t,x}^{\mathbb{P}^{\Gamma^*}} \left[ \theta \int_{t}^{T} g(s,X_s,h_s,\gamma^*_s;\theta) \, ds \right]
    \end{align}
    holds for any admissible strategies $H = \left(h_s\right)_{s\in[t,T]} \in \mathcal{A}^H_T$ and $\Gamma = \left(\gamma_s\right)_{s \in [t,T]} \in \mathcal{A}^\Gamma_T$, and that
    \begin{align}
            u(t,x) =\mathbf{E}_{t,x}^{\mathbb{P}^{\Gamma^*}} \left[ \theta \int_{t}^{T} g(s,X_s,h^*_s,\gamma^*_s;\theta) \, ds \right].
    \end{align} 
\end{theorem}

In summary, the results in Section \ref{sec:mainresult} provide the finite-horizon duality foundation for the learning methodology developed below. The free energy--entropy duality identifies the risk-sensitive control problem with an LQG stochastic differential game; the quadratic form of the value function motivates the critic parameterization; and the affine form of the saddle-point policies motivates the actor parameterization. The reinforcement-learning approach developed in Section~\ref{sec:PolicyGradientMethods} preserves this structure while replacing the closed-form solution of the known-coefficient model by data-driven learning procedures.

\subsection{The Kelly Criterion}
\label{sec:Kelly}

The Kelly criterion corresponds to the risk-neutral limiting case of the risk-sensitive criterion as $\theta \to 0$ and also coincides with the logarithmic utility. It is of interest as both a limiting case of the duality-based solution above and as an interpretable reference policy for the reinforcement-learning approach that will be developed in Section~\ref{sec:PolicyGradientMethods}. For $R_t=0$, the Kelly criterion is
\begin{align}\label{eq:J:Kelly}
J^K_T(t,x;H)
:=
\mathbf E_{t,x}\left[R_T\right]
=
\mathbf E_{t,x}
\left[
\int_t^T
\left\{
-\frac12 h_s'\Sigma_s\Sigma_s'h_s
+h_s'a_s
+\frac12\Xi_s'\Xi_s
-c_s
+
\left(h_s'A_s-C_s\right)X_s
\right\}
\, ds
\right],
\end{align}
where the equality follows from \eqref{eq:excess_return} and from the martingale property of the stochastic integral
$\int_t^T
\left(h_s'\Sigma_s-\Xi_s'\right)dW_s$. Define the Kelly value function by
\begin{align}\label{eq:valuefunction:u:Kelly}
\Phi^K(t,x)
:=
\sup_{H\in\mathcal{A}^H_T}
J^K_T(t,x;H)
= 
\sup_{H\in\mathcal{A}^H_T}
\mathbf E_{t,x}
\left[
\int_t^T
g^K(s,X_s,h_s) \, ds
\right],
\end{align}
where
\begin{align}
g^K(s,x,h)
:=
-\frac12 h'\Sigma_s\Sigma_s'h
+h'a_s
+\frac12\Xi_s'\Xi_s
-c_s
+
\left(h'A_s-C_s\right)x.
\label{eq:g:Kelly}
\end{align}
The state process evolves under $\mathbb P$ according to \eqref{eq:state}. Hence, $\Phi^K$ solves the HJB equation
\begin{align}
\frac{\partial \Phi^K(s,x)}{\partial s}
+ \left(b_s+B_sx\right)'D\Phi^K(s,x)
+ \frac{1}{2}
\tr\left(\Lambda_s\Lambda_s'D^2\Phi^K(s,x)\right)
+ \sup_{h\in\mathbb R^m} g^K(s,x,h)
=
0,
\label{eq:HJB:Kelly}
\end{align}
with terminal condition $\Phi^K(T,x)=0$.

Because the state process is uncontrolled under $\mathbb P$, the optimal control is obtained by pointwise maximization of $g^K$. Assumption~\ref{as:sigma:posdef} gives
\begin{align}
h^K(s,x)
= 
\left(\Sigma_s\Sigma_s'\right)^{-1}
\left(a_s+A_sx\right),
\qquad
(s,x)\in[t,T)\times\mathbb R^n.
\label{eq:hhat:Kelly}
\end{align}

The value function can be shown to be quadratic, i.e. 
$
\Phi^K(t,x)
= 
\frac12 x'Q_t^Kx + \left(q_t^K\right)'x+k_t^K,
$
with coefficients obtained by substituting \eqref{eq:hhat:Kelly} into the HJB equation. This Kelly solution provides a $\theta \to 0$ reference point against which the risk-sensitive policies and the learning-based approximations can be compared.

\section{Extension to Infinite-Horizon Problems}
\label{sec:RSBAM:infinite}

Section~\ref{sec:RSBAM} provides the finite-horizon duality foundation for reinforcement-learning algorithms formulated as \emph{episodic} tasks. To support \emph{continuing} tasks, this section extends the duality-based solution to an infinite-horizon risk-sensitive control problem. Under constant coefficients and stabilizing assumptions, the finite-horizon Riccati system admits limiting coefficients, which yield a relative value function, an ergodic Bellman--Isaacs equation, and stationary saddle-point policies. The section adapts the asymptotic results of \citet{dall_RSBench}, originally derived under a Kuroda--Nagai change of measure, to the limiting stochastic differential games produced by the free energy--entropy duality. This transfer is justified by the equivalence between the two approaches established in \citet[Section~3]{LleoRunggaldier_RSIMviaDuality_2026}. The section also connects the ergodic game value to the long-run risk-sensitive criterion and identifies the ergodic Kelly investment policy as the risk-neutral limiting reference case.

\begin{assumption}\label{as:coefficients:ergodic}
    The coefficient functions in \eqref{eq:dS}--\eqref{eq:dL} are constant.    
\end{assumption}

\begin{definition}[Admissible (Target) Policies]\label{def:class:barA}
\;
\begin{enumerate}[(i)]
    \item $H = \left(h_t\right)_{t \geq 0}$ with values in $\mathbb{R}^{m}$ is in class $\bar{\mathcal{A}}^{H}$ if $H\big|_{[0,T]} \in \mathcal{A}^H_T$ for every $T >0$.

    \item $\Gamma = \left(\gamma_t\right)_{t \geq 0}$ with values in $\mathbb{R}^{d}$ is in class $\bar{\mathcal{A}}^{\Gamma}$ if $\Gamma\big|_{[0,T]} \in \mathcal{A}^\Gamma_T$ for every $T >0$, and the mean $m_t$ and covariance matrix $P_t$ of the state process $X_t$ under $\mathbb{P}^\Gamma$ admit finite, nondegenerate limits as $t\to\infty$.
\end{enumerate}

Such strategies are called admissible for the infinite horizon problem.
\end{definition}

The investor seeks to maximize the risk-sensitive log excess return per unit time:
\begin{equation}\label{eq:J:infinite}
J_\infty(H,\theta)
:=
\liminf_{T\to\infty}
-\frac{1}{\theta T}
\ln I_T(H,\theta).
\end{equation}
where the criterion $I_T$ is defined at \eqref{eq:I}. The ergodic criterion $J_\infty$ at \eqref{eq:J:infinite} is the long-run counterpart of the finite horizon criterion $J_T$ at \eqref{eq:J}. Equivalently, since $\theta>0$,
$
J_\infty(H,\theta)
=
-\frac1\theta
\limsup_{T\to\infty}
\frac{1}{T}
\ln I_T(H,\theta)
$.
We therefore define the exponentially transformed ergodic criterion by
\begin{align}\label{eq:I:infinite}
I_\infty(H,\theta)
&:=
e^{-\theta J_\infty(H,\theta)}
=
\exp\left\{
\limsup_{T\to\infty}
\frac{1}{T}
\ln I_T(H,\theta)
\right\}
=
\limsup_{T\to\infty}
\left(I_T(H,\theta)\right)^{1/T}.
\end{align}
Thus, for $\theta>0$, maximizing $J_\infty(H,\theta)$ is equivalent to minimizing $I_\infty(H,\theta)$.

Motivated by the finite-horizon dual representation, we consider the ergodic stochastic differential game
\begin{align}\label{eq:SDG:ergodic}
    \bar u(x) := \inf_{H \in \bar{\mathcal{A}}^H} \sup_{ \Gamma \in \bar{\mathcal{A}}^\Gamma} \limsup_{T \to \infty} \frac{1}{T}\mathbf{E}_x^{\mathbb{P}^\Gamma} \left[ \theta \int_{0}^{T} g(X_t,h_t,\gamma_t;\theta) \, dt
    \right], 
\end{align}
where $\bar u$ is the relative value function. We solve this game in the next two subsections, and then show that this game corresponds to an ergodic extension of the duality-induced relation \eqref{eq:EEDuality:inf:SDG}.

\subsection{Asymptotic Behavior of the Value Function Coefficients}

For notational convenience, let
\begin{align}
    \mathcal{P}^{-} :=& I_d - \frac{\theta}{\theta+1}\Sigma'\left(\Sigma\Sigma'\right)^{-1}\Sigma
                        \label{eq:notation:calPmin}\\
    K_0 :=& \theta \Lambda \mathcal{P}^{-} \Lambda'
                            \label{eq:notation:K0}\\  
    K_1 :=& B - \frac{\theta}{\theta+1} \Lambda\Sigma'(\Sigma\Sigma')^{-1}A
                            \label{eq:notation:K1}
\end{align}
The Riccati equation~\eqref{eq:Q:Riccati} becomes
\begin{align}\label{eq:Q:Riccati:v2}
    \dot{Q}_t 
    - Q_t K_0 Q_t
    + K_1'Q_t 
    + Q_t K_1 
    + \frac{1}{\theta+1} A' (\Sigma\Sigma')^{-1} A
    = 0,    \quad Q_T = 0,
\end{align}
with an analogous reformulation for the linear ODE \eqref{eq:q:ODE} and integral \eqref{eq:k:integral}.

Because the Riccati system coincides with that obtained under the Kuroda--Nagai change of measure, the asymptotic results of \citet{dall_RSBench} apply directly. 

\begin{proposition}[Proposition 1 in \citep{dall_RSBench}]\label{prop:Qqk:asymptotic}
\;
\begin{enumerate}[(i).]
    \item If 
    \begin{align}\label{eq:stabcond:K1}
        K_1 \text{\, defined at \,} \eqref{eq:notation:K1} \text{\, is stable},
    \end{align}
    then the finite-horizon Riccati coefficient $Q$ converges, as the horizon $T\to+\infty$, to a nonnegative definite matrix $\bar{Q}$, which is a solution of the algebraic Riccati equation:
    \begin{align}\label{eq:Q:Riccati:algebraic}
        - \bar{Q}K_0\bar{Q} 
        + K_1'\bar{Q} +\bar{Q}K_1
        + \frac{1}{\theta+1} A' (\Sigma\Sigma')^{-1} A
        = 0
    \end{align}
    Moreover, $\bar{Q}$ satisfies the estimate
    \begin{align}\label{eq:Qbar:Riccati:Infhorizon}
        0 \leq \bar{Q} \leq
        \frac{1}{\theta}\int_{0}^{+\infty}
        e^{s K_1'}A'(\Sigma\Sigma')^{-1}Ae^{s K_1} \, ds
    \end{align}
    \item In addition, $q$ converges as $T \to +\infty$ to a constant vector $\bar{q}$, which satisfies
    \begin{align}\label{eq:qbar:ODE:Infhorizon}
        \left( K_1' - \bar{Q}K_0 \right)\bar{q}
        + \bar{Q}b + \theta \bar{Q}\Lambda\Xi - C'
        +\frac{1}{\theta+1} \left(A'-\theta \bar{Q}\Lambda\Sigma'\right)(\Sigma\Sigma')^{-1}\left(a +\theta\Sigma\Xi \right) = 0,
    \end{align}
    and $\frac{k}{T}$ converges to a constant $\bar{k}$ defined by
    \begin{align}\label{eq:kbar:integral:Infhorizon}
        \bar{k} 
        =& 
        \frac{1}{2} \textrm{tr} \left( \Lambda \Lambda' \bar{Q} \right)
        - \frac{\theta}{2}\bar{q}'\Lambda\Lambda'\bar{q} 
        + b'\bar{q}
        +\frac{1}{2}\frac{1}{\theta+1} a'(\Sigma\Sigma')^{-1}a
        +\frac{1}{2}\frac{\theta^2}{\theta+1}\bar{q}'\Lambda\Sigma'(\Sigma\Sigma')^{-1}\Sigma\Lambda' \bar{q}
                                        \nonumber\\
        &-\frac{\theta}{\theta+1}\bar{q}'\Lambda\Sigma'(\Sigma\Sigma')^{-1}a
        -\frac{\theta^2}{\theta+1} \bar{q}'\Lambda\Sigma'(\Sigma\Sigma')^{-1}\Sigma\Xi
        +\theta\Xi'\Lambda'\bar{q}
        -\frac{1}{2}\left(\theta-1\right)\Xi'\Xi
                                        \nonumber\\
        &+\frac{\theta}{\theta+1} a'(\Sigma\Sigma')^{-1}\Sigma\Xi
        +\frac{1}{2}\frac{\theta^2}{\theta+1}\Xi'\Sigma'(\Sigma\Sigma')^{-1}\Sigma\Xi
        -c
    \end{align}
    \item If, in addition to condition~\eqref{eq:stabcond:K1} we assume that
    \begin{equation}\label{eq:ctrlcond:BA}
        (B',A'(\Sigma\Sigma')^{-1}\Sigma)
        \textrm{\, is controllable}
    \end{equation}
    then the solution $\bar{Q}$ of \eqref{eq:Q:Riccati:algebraic} is strictly positive definite. 
\end{enumerate}
\end{proposition}

\subsection{Ergodic Bellman--Isaacs Equation}

The ergodic Bellman--Isaacs PDE is
\begin{align}\label{eq:BellmanIsaacs:ergodic}
    \mathcal{H}\left(x,D\bar u(x), D^2 \bar u(x)\right) = \rho, 
\end{align}
where 
\begin{align}\label{eq:Hamiltonian:Hpm}
   \mathcal{H}\left(x,D\bar u(x), D^2 \bar u(x)\right)
   :=& \inf_{h \in \mathbb{R}^m} \sup_{\gamma \in \mathbb{R}^d} \left\{ \mathcal{L}^{h,\gamma}\bar u(x) +\theta g(x, h,\gamma;\theta) \right\}
   = \sup_{\gamma \in \mathbb{R}^d} \inf_{h \in \mathbb{R}^m} \left\{ \mathcal{L}^{h,\gamma}\bar u(x) +\theta g(x,h,\gamma;\theta) \right\},
\end{align}
the generator $\mathcal{L}^{h,\gamma}\bar u$ of $\bar u(X_t)$, assuming $\bar u \in C^2(\mathbb{R}^n)$, is
\begin{align}\label{eq:operator:call}
    \mathcal{L}^{h,\gamma}\bar u(x) 
   :=& \left[b + B x + \Lambda \gamma \right]' D\bar u(x)
    + \frac{1}{2} \tr \left(\Lambda\Lambda' D^2\bar u(x) \right).
\end{align}
The constant $\rho$ is the optimal long-run value of the dual stochastic differential game. Equivalently, when the limiting duality relation holds, the optimal long-run risk-sensitive criterion is $-\rho/\theta$.

The ergodic Bellman--Isaacs PDE at \eqref{eq:BellmanIsaacs:ergodic} has a solution given by the pair
\begin{align}
    \bar u(x) =& -\theta \left( \frac{1}{2} x' \bar{Q} x + \bar{q}'x \right),
                    \label{eq:baru}\\
    \rho =& -\theta \bar{k},
                    \label{eq:rho}
\end{align}
where $\bar{Q}, \bar{q}$, and $\bar{k}$ are given in Proposition \ref{prop:Qqk:asymptotic}.

\begin{theorem}\label{theo:ergodic}

Let $\bar u$ be the relative value function at \eqref{eq:baru}, and let $\rho$ be the constant defined at \eqref{eq:rho}. Define the stationary strategies $H^* = \left(h^*(X_t) \right)_{t \geq 0}$ and $\Gamma^* = \left(\gamma^*(X_t) \right)_{t \geq 0}$, by $h^*(x) = \bar{h}(x,D\bar u(x))$ and $\gamma^*(x) = \bar{\gamma}(x,D\bar u(x))$, where $\bar{h}(\cdot)$ and $\bar{\gamma}(\cdot)$ are given by
\begin{align}
        \bar{h}(x,D\bar u(x))
            =&  \frac{1}{\theta+1}\left(\Sigma\Sigma'\right)^{-1}
        \left( a + A x
        + \theta \Sigma \Xi 
        + \Sigma\Lambda' D\bar u(x)
        \right)
                            \label{eq:hbar}\\    
    \bar{\gamma}(x,D\bar u(x)) 
=& \Lambda' D\bar u(x) - \theta \left(\Sigma'\bar{h}(x,D\bar u(x)) - \Xi \right)
                        \label{eq:gammabar:1}\\
    \phantom{\bar{\gamma}(x,D\bar u(x))}
=&  \left[ I_d- \frac{\theta}{\theta+1} \Sigma'\left(\Sigma\Sigma'\right)^{-1}\Sigma \right]\Lambda' D\bar u(x) 
        - \frac{\theta}{\theta+1} \Sigma'\left(\Sigma\Sigma'\right)^{-1}
        \left( a + A x \right) 
                                    \nonumber\\
        &+ \theta \left[ I_d- \frac{\theta}{\theta+1} \Sigma'\left(\Sigma\Sigma'\right)^{-1}\Sigma \right] \Xi,
                        \label{eq:gammabar:2}
    \end{align}   
    or equivalently
     \begin{align}
        \bar{h}(x,D\bar u(x))
        =&  \left(\Sigma\Sigma'\right)^{-1}  
        (a + A x) + \left(\Sigma\Sigma'\right)^{-1}\Sigma \bar{\gamma}(x,D\bar u(x))
                             \label{eq:hbar:alt}\\
        \bar{\gamma}(x,D\bar u(x)) 
=& \left[
            I_d+ \theta \Sigma' \left(\Sigma\Sigma'\right)^{-1}\Sigma
        \right]^{-1} 
        \left[ 
            - \theta \Sigma' \left(\Sigma\Sigma'\right)^{-1} (a + A x)
            + \theta \Xi
            + \Lambda'D\bar u(x)
        \right].
                        \label{eq:gammabar:alt}
\end{align}
If the assumptions in Proposition \ref{prop:Qqk:asymptotic}, namely \eqref{eq:stabcond:K1} and \eqref{eq:ctrlcond:BA}, hold, and we have both
\begin{align}
    & (B',\Lambda)
    \quad \textrm{is controllable}
            \label{eq:ctrlcond:BA:theorem}\\
    \theta & \bar{Q}\Lambda\Sigma'(\Sigma\Sigma')^{-1}\Sigma\Lambda'\bar{Q}
    < A'(\Sigma\Sigma')^{-1} A,
            \label{eq:theo:ergodic:mainassumption}
\end{align}
then,
\begin{enumerate}[1.]
    \item The pair of strategies $(H^*,\Gamma^*)$ is admissible, with $H^*\in\bar{\mathcal A}^H$ and $\Gamma^*\in\bar{\mathcal A}^\Gamma$. Moreover, their defining policies $h^*(\cdot)$ and $\gamma^*(\cdot)$ are Borel-measurable maps.

       \item The saddle point inequalities
        \begin{align}\label{eq:verification:saddle:ergodic}
            \limsup_{T \to \infty} \frac{1}{T} \mathbf{E}_{x}^{\mathbb{P}^\Gamma} \left[ \theta \int_{0}^{T} g(X_t,h^*_t,\gamma_t;\theta) \, dt \right]
            \leq 
            \rho 
            \leq 
            \limsup_{T \to \infty} \frac{1}{T}
            \mathbf{E}_{x}^{\mathbb{P}^{\Gamma^*}} \left[ \theta \int_{0}^{T} g(X_t,h_t,\gamma^*_t;\theta) \, dt \right]
        \end{align}
        holds for any admissible strategies $H = \left(h_t\right)_{t \geq 0} \in \bar{\mathcal{A}}^H$ and $\Gamma = \left(\gamma_t\right)_{t \geq 0} \in \bar{\mathcal{A}}^\Gamma$, and
        \begin{align}
            u(x) = \limsup_{T \to \infty} \frac{1}{T} \mathbf{E}_{x}^{\mathbb{P}^{\Gamma^*}} \left[ \theta \int_{0}^{T} g(X_t,h^*_t,\gamma^*_t;\theta) \, dt \right] = \rho.
        \end{align} 
    
    \item The strategies $H^*$ and $\Gamma^* $ are optimal for the ergodic stochastic differential game \eqref{eq:SDG:ergodic}.

\end{enumerate}
\end{theorem}

\begin{proof}{Proof}
See Appendix \ref{app:proof:theo:ergodic}.
\end{proof}

\subsection{Relation between the Ergodic Stochastic Differential Game and the Free Energy--Entropy Duality}

The preceding theorem identifies the long-run value of the ergodic stochastic differential game induced by the free energy--entropy duality. Under the assumptions of Theorem~\ref{theo:ergodic}, the stationary saddle point $(H^*,\Gamma^*)$ yields the long-run value
\begin{align}\label{eq:ergodic:value:rho}
\rho
=
\inf_{H \in \bar{\mathcal A}^H}
\sup_{\Gamma \in \bar{\mathcal A}^\Gamma}
\limsup_{T\to\infty}
\frac1T
\mathbf E_x^{\mathbb P^\Gamma}
\left[
\theta\int_0^T
g(X_t,h_t,\gamma_t;\theta)\,dt
\right].
\end{align}
Consequently, the limiting duality relation gives
$\inf_{H\in\bar{\mathcal A}^H}
I_\infty(H,\theta)
=
e^\rho$.
Equivalently, since $I_\infty(H,\theta)=e^{-\theta J_\infty(H,\theta)}$ and $\theta>0$, the optimal long-run risk-sensitive criterion is
\begin{align}\label{eq:Jinfty:rho}
\sup_{H\in\bar{\mathcal A}^H}
J_\infty(H,\theta)
=
-\frac{\rho}{\theta}.
\end{align}

\subsection{The Ergodic Kelly Criterion}
\label{sec:Kelly:ergodic}

The ergodic Kelly criterion is the risk-neutral limiting case of the ergodic risk-sensitive criterion. For a policy $H$, define
\begin{align}\label{eq:J:Kelly:ergodic}
J_\infty^K(x;H)
:=
\liminf_{T\to\infty}
\frac{1}{T}
\mathbf{E}_x\left[R_T-R_0\right].
\end{align}
Under constant coefficients, the corresponding relative value function $\Phi_\infty^K(x) :=
\sup_{H\in\bar{\mathcal A}^H}
J_\infty^K(x;H)$ solves the following ergodic Bellman--Isaacs PDE (see \eqref{eq:BellmanIsaacs:ergodic}-\eqref{eq:Hamiltonian:Hpm}):
\begin{align}\label{eq:HJB:Kelly:ergodic}
\rho_K
=
(b+Bx)'D\Phi_\infty^K(x)
+ \frac{1}{2}\tr\left(\Lambda\Lambda'D^2\Phi_\infty^K(x)\right)
+ \sup_{h\in\mathbb R^m} g^K(x,h).
\end{align}

The pointwise maximizer is
\begin{align}\label{eq:hK:ergodic}
h^K(x)
=
(\Sigma\Sigma')^{-1}(a+Ax).
\end{align}
The relative value function is quadratic,
$\Phi_\infty^K(x)
=
\frac{1}{2} x'\bar Q^Kx + \left(\bar q^K\right)'x$,
where $\bar Q^K,\bar q^K$, and $\rho_K$ solve the corresponding algebraic Kelly equations obtained from \eqref{eq:HJB:Kelly:ergodic}.

\section{Actor--Critic Methods for Continuous-Time LQG Stochastic Differential Games}
\label{sec:PolicyGradientMethods}

Sections~\ref{sec:RSBAM} and~\ref{sec:RSBAM:infinite} provide model-based finite- and infinite-horizon solutions to the duality-induced stochastic differential games. Building on this foundation, we develop reinforcement learning algorithms. We retain the LQG specification as a modeling framework, but no longer assume that the coefficients are known. 

We focus on continuous-time $q$-learning actor--critic methods, following the martingale characterization of \citet{jiaPolicyEvaluationTemporalDifference2022,jiaPolicyGradientActor2022,jiaQLearningContinuousTime}. Relative to that general framework, our setting is specialized to the LQG stochastic differential game generated by the benchmarked investment problem. The quadratic form of the value functions motivates the critic parametrization, while the affine form of the saddle-point policies motivates the actor parametrization. The actors are deterministic Markov policies in the sense of \citet{silverDeterministicPolicyGradient2014}: they map the time--state pair directly into the portfolio control and the adversarial control.

The finite-horizon game of Section~\ref{sec:RSBAM} leads to an \emph{episodic} actor--critic algorithm, in which each episode corresponds to one run of the model over $[0,T]$. The ergodic game of Section~\ref{sec:RSBAM:infinite} leads instead to a \emph{continuing} actor--critic algorithm, in which learning proceeds along a non-terminating trajectory. This section develops the \emph{episodic} actor critic--algorithm in detail. The \emph{continuing} algorithm follows with some minor adaptations. In both cases, the partly model-based structure reduces the burden of function approximation and provides interpretable reference policies, including the Kelly policy discussed in Subsection~\ref{sec:PolicyGradientMethods:Explainability}.

\subsection{Preliminary Results for Episodic Tasks}\label{sec:PG:Preliminary}
\subsubsection{State-Value Function $v$}

Recall that the admissible strategies in $\mathcal{A}^H_T$ and $\mathcal{A}^\Gamma_T$ are
given by Markov control processes. This allows us to define the value function $v$ associated with the strategies $H = (h_s)_{s \in [t,T]} \in\mathcal{A}^H_T$, $\Gamma = (\gamma_s)_{s \in [t,T]} \in\mathcal{A}^\Gamma_T$ as the $\mathbb{P}^\Gamma$-expectation 
\begin{align}\label{eq:valuefunc:v}
    v(t,x;H,\Gamma,\theta)   
    := \mathbf{E}_{t,x}^{\mathbb{P}^\Gamma} \left[ \theta \int_{t}^{T} g(s,X_s,h_s,\gamma_s;\theta) \, ds \right],
\end{align}
where $g$ is defined at \eqref{eq:g}. Hence, the \emph{optimal value function} $u$ defined at~\eqref{eq:valuefunction:u} is 
\begin{align}\label{eq:vf:u}
    u(t,x) = \inf_{H \in \mathcal{A}^H_T} \sup_{ \Gamma \in \mathcal{A}^\Gamma_T} v(t,x;H,\Gamma,\theta).
\end{align} 

\begin{definition}[on-and off-policy learning]\label{def:on-off-policy}
    \emph{On-policy} learning optimizes the \emph{target} policy directly. \emph{Off-policy} learning optimizes the \emph{target} policy based on a \emph{behavior} policy (see Definition~\ref{def:classA:behavior} below). 
\end{definition}

\begin{remark}
The change-of-measure representation of the state-value function $\nu$ can be viewed as an application of importance sampling with $\chi^{\Gamma}_{[0,T]}$ at \eqref{eq:RNderivative:gamma} acting as the likelihood ratio. In reinforcement learning, this structure appears in the \emph{off-policy} approach. Our approach performs \emph{on-policy} learning while still relying on importance sampling in the probabilistic sense to handle the risk-sensitive change of measure.
\end{remark}

By the Feynman--Kac formula, for fixed policies $H$ and $\Gamma$, the value function $v$ solves the parabolic PDE
\begin{align}\label{eq:PDEcharacterization:v}
    \frac{\partial v}{\partial t}
    + \left[b_s + B_s x + \Lambda_s \gamma(s,x) \right]' \frac{\partial v}{\partial x}
    + \frac{1}{2} \tr \left(\Lambda_s\Lambda_s' \frac{\partial^2 v}{\partial x^2} \right)
    +\theta g(s,x,h(s,x),\gamma(s,x);\theta) 
    = 0, \; (t,x) \in [0,T) \times \mathbb{R}^n
\end{align}
with terminal condition $v(T,x;H,\Gamma,\theta) = 0, \; \forall x \in \mathbb{R}^n$. The optimal value function $u(t,x)$ satisfies the Bellman--Isaacs PDE at~\eqref{eq:BellmanIsaacs:u}-\eqref{eq:Hamiltonian:H+}, that is
\begin{align}\label{eq:PDEBI:u}
    \frac{\partial u}{\partial t}
    + \inf_{h \in \mathbb{R}^m} \sup_{ \gamma  \in \mathbb{R}^d}\left\{\left[b_s + B_s x + \Lambda_s \gamma \right]' \frac{\partial u}{\partial x}
    + \frac{1}{2} \tr \left(\Lambda_s\Lambda_s' \frac{\partial^2 u}{\partial x^2} \right)
    +\theta g(s,x,h,\gamma;\theta) \right\}
    = 0,
\end{align}
with terminal condition $u(T,x) = 0, \; \forall x \in \mathbb{R}^n$.

Next, for a filtration $\left(\mathcal{F}_s\right)_{s \in [0,T]}$, semimartingale $(Y_s)_{s \geq 0}$, and measure $\mathbb{P}^\Gamma$ defined at \eqref{eq:RNderivative:gamma}, we denote by
\begin{align}\label{eq:HilbertSpace:LF2}
    L_\mathcal{F}^2 \left([t,T];Y;\mathbb{P}^\Gamma\right) 
    :=& \Bigg\{(\kappa_s)_{ s \in [t,T]}: \kappa_s \text{ is } \mathcal{F}_s\text{-progressively measurable, }
    \mathbf{E}^{\mathbb{P}^\Gamma} \left[ \int_t^T \lvert \kappa_s \rvert^2 d \langle Y \rangle_s \right] < \infty \Bigg\} 
\end{align}
the Hilbert space with $L^2$-norm $\| \kappa \|_{L^2} := \left( \mathbf{E}^{\mathbb{P}^\Gamma}\left[ \int_t^T \lvert \kappa_s \rvert^2 d \langle Y \rangle_s \right]\right)^{1/2}$, 
where $\langle \cdot \rangle$ denotes quadratic variation. We now introduce the following martingale process.

\begin{definition}\label{def:Mt}
    For given admissible strategies $H \in \mathcal{A}^H_T, \Gamma \in \mathcal{A}^{\Gamma}_T$, let $\left(M_s\right)_{s \in [t,T]}$ be the $\mathbb{P}^\Gamma$-process
    \begin{align}\label{eq:Martingale:M}
        M_s := v(s,X_s;H,\Gamma,\theta)
        + \theta \int_t^s g(u,X_u,h_u,\gamma_u;\theta) du,
    \end{align}
    where  $v$ is the function defined at \eqref{eq:valuefunc:v},  $g$ is defined at \eqref{eq:g}, and $\left(X_s\right)_{s \in [t, T]}$ satisfies \eqref{eq:state:Pgamma:FO} under $\mathbb{P}^\Gamma$ with an initial value that is either fixed, that is $X_t = x \in \mathbb{R}^n$, or random, meaning here that $X_t = x \sim N(\mu, \Lambda_0)$ and is independent of the Wiener process $W_s$.
\end{definition}

Proposition 1 in \citet{jiaPolicyEvaluationTemporalDifference2022} shows that under standard assumptions, the process $\left(M_s\right)_{s \in [t,T]}$ is a square-integrable martingale uniquely defined by the value function $v$. We restate their result in our setting and extend it to allow a random initial value for the process $\left(X_s\right)_{s \in [t,T]}$:

\begin{proposition}\label{prop:M:mart}
    Under the assumptions on the functions $b,B,\Lambda,a,A,\Sigma,c,C,\Xi$ presented in Section~\ref{sec:RSBAM:model}, the process $\left(M_s\right)_{s \in [t,T]}$ introduced in Definition \ref{def:Mt} is a square-integrable $\mathbb{P}^\Gamma$-martingale. 

    Additionally, the process $\left(M_s\right)_{s \in [t,T]}$ is unique. Formally, if there is a continuous function $\check{v}$ such that 
    \begin{enumerate}[(i)]
        \item $\check{v}(T,x) = 0, \; \forall x \in \mathbb{R}^n$, and
        \item the process $\check{M}_s := \check{v} \left(s,X_s\right) + \theta \int_t^s g(u,X_u,h_u,\gamma_u;\theta) du$ is a square-integrable $\mathbb{P}^\Gamma$-martingale, 
    \end{enumerate}
    then $\check{v} = v$ on $[0,T] \times \mathbb{R}^n$. 
\end{proposition}

\begin{proof}{Proof}
    The proof follows along the lines of that of Proposition 1 in \citet{jiaPolicyEvaluationTemporalDifference2022}. The extension to a random initial value follows immediately from the assumption that $X_0$ is independent of $W_s$. 
\end{proof}

\begin{corollary}[Martingale Orthogonality Condition]\label{coro:M:mart:orthogonality}
Let $(M_s)_{s \in [t,T]}$ be the square-integrable $\mathbb{P}^\Gamma$-martingale from Definition \ref{def:Mt}, for fixed $H \in \mathcal{A}^H_T, \Gamma \in \mathcal{A}^\Gamma_T$, and given the  fixed or random initial condition $X_t=x$. Then, for any test function $\xi \in L_\mathcal{F}^2 \left([t,T];M;\mathbb{P}^\Gamma\right)$, 
\begin{align}\label{eq:MartOrthoCond}
    \mathbf{E}_{t,x}^{\mathbb{P}^\Gamma} \left[ \int_t^s \xi_u dM_u\right] 
    = \mathbf{E}_{t,x}^{\mathbb{P}^\Gamma} \left[ \int_t^s \xi_u \left\{ dv \left(u,X_u; H, \Gamma, \theta\right)
    + \theta g(u,X_u,h_u,\gamma_u;\theta) du\right\}\right] 
    = 0.
\end{align}
\end{corollary}

\begin{remark}
    Jia and Zhou noted that the term `test function' is used by convention; $\xi_s$ is in fact a stochastic process. Hence, the corollary follows from the properties of martingales. 
\end{remark}

\subsubsection{Action-Value Function $q$}

So far, we have discussed the state-value function $v$. The action-value function $q$ also plays a pivotal role in deterministic policy gradient algorithms \citep[see][]{silverDeterministicPolicyGradient2014}. Consistently with \citet{jiaQLearningContinuousTime}, we now define the function $q$ as the ``left-hand side'' of the PDE \eqref{eq:PDEcharacterization:v}:
\begin{align}\label{eq:q:func}
    q\left( s,x,h,\gamma;\theta\right)
    := \frac{\partial v}{\partial t}
    + \left[b_s + B_s x + \Lambda_s \gamma \right]' \frac{\partial v}{\partial x} 
    + \frac{1}{2} \tr \left(\Lambda_s\Lambda_s' \frac{\partial^2 v}{\partial x^2} \right)
    +\theta g(s,x,h,\gamma;\theta).
\end{align}
Following \citet{jiaQLearningContinuousTime}, we call this object the continuous-time action-value function, although analytically it is the Hamiltonian residual associated with the value function. These authors showed that $q$ is the continuous-time analog to the action-value function $Q$ in discrete time.

We now present a proposition motivated by the martingale characterization in Jia and Zhou, and which adapts the logic of their Theorem 7 to deterministic Markov policies. A key difference between the deterministic policies employed in this paper and the exploratory policies considered in \citet{jiaQLearningContinuousTime} is that deterministic policies do not themselves identify the action dependence of $q$ \citep{silverDeterministicPolicyGradient2014}. Importantly, this result addresses both target policies, policies we aim to optimize, and behavior policies, policies used for evaluation.  The class of behavior policies we work with is the following:

\begin{definition}[Behavior Policy] \label{def:classA:behavior}
A \emph{behavior investment policy} $H=(h_s)_{s\in[0,T]}$ belongs to $\mathcal A_T^{H,\mathrm{beh}}$ if $h_s$ is $\mathbb R^m$-valued, progressively measurable, c\`adl\`ag, and square-integrable on $[0,T]$.
A \emph{behavior adversarial policy} $\Gamma=(\gamma_s)_{s\in[0,T]}$ belongs to $\mathcal A_T^{\Gamma,\mathrm{beh}}$ if $\gamma_s$ satisfies the same conditions and $\chi^\Gamma_{[0,T]}$ is an exponential martingale.
\end{definition}

\begin{remark}
We will use a behavior policy defined as a perturbed version of the target policy in the implementation in Section \ref{sec:Implementation} below. We therefore implicitly assume that the filtration $\left(\mathcal F_s \right)_{0\leq s\leq T}$ tracks both the Brownian motion $W$ and the independent perturbation. Note as well that admissible policies, the set of policies in which we seek the optimal policy,  are a subset of behavior policies.    
\end{remark}

Under the regularity and uniqueness conditions of \citet[Theorem~7]{jiaQLearningContinuousTime}, the following martingale characterization specializes to the present LQG game. We use the superscripts $\mathrm{tar}$ and $\mathrm{beh}$ to distinguish target and behavior policies.   

\begin{proposition}\label{prop:qlearning:Martingale_condition}
Fix target policies $H^\mathrm{tar}\in\mathcal A_T^H$ and $\Gamma^\mathrm{tar}\in\mathcal A_T^\Gamma$. Let
$\tilde v:[t,T]\times\mathbb R^n\to\mathbb R$,
$\tilde q:[t,T]\times\mathbb R^n\times\mathbb R^m\times\mathbb R^d\to\mathbb R$
be sufficiently smooth, with $\tilde v(T,x)=0$ and
\begin{align}\label{eq:prop:qlearning:Martingale_condition:tildeq}
\tilde q(s,x,h_s^\mathrm{tar},\gamma_s^\mathrm{tar};\theta)=0,
\qquad
(s,x)\in[t,T]\times\mathbb R^n,
\end{align}
without $\tilde q$ being identically zero. Then $\tilde v$ and $\tilde q$ are the value function and $q$-function associated with
$(H^\mathrm{tar},\Gamma^\mathrm{tar})$ if and only if
\begin{align}\label{eq:qlearning:Martingale:M}
    \tilde{M}_s 
    :=  \tilde{v} \left(s,X_s\right) 
    +   \int_t^s \left[ \theta g(u,X_u,h^\mathrm{tar}_u, \gamma^\mathrm{tar}_u;\theta) - \tilde{q}(u,X_u,h^\mathrm{tar}_u, \gamma^\mathrm{tar}_u;\theta) \right]du ,
\end{align}
is a $\mathbb P^{\Gamma^\mathrm{tar}}$-martingale.

Moreover, if $\tilde v$ and $\tilde q$ are the corresponding value and $q$-functions, then for any behavior policies
$H^\mathrm{beh}\in\mathcal A_T^{H,\mathrm{beh}}$ and $\Gamma^\mathrm{beh}\in\mathcal A_T^{\Gamma,\mathrm{beh}}$, the process
\begin{align}\label{eq:qlearning:Martingale:M:check}
    \tilde{M}_s 
    :=  \tilde{v} \left(s,X_s\right) 
    +   \int_t^s \left[ \theta g(u,X_u,h^\mathrm{beh}_u, \gamma^\mathrm{beh}_u;\theta) - \tilde{q}(u,X_u,h^\mathrm{beh}_u, \gamma^\mathrm{beh}_u;\theta) \right]du ,
\end{align}
is a $\mathbb P^{\Gamma^\mathrm{beh}}$-martingale. Conversely, the existence of such a behavior-policy martingale identifies $\tilde v$ and $\tilde q$. Finally, if the target policies satisfy the pointwise saddle-point conditions
\begin{align}\label{eq:h_gamma:argmin}
h_s^\mathrm{tar}\in
\argmin_h \tilde q(s,x,h,\gamma_s^\mathrm{tar};\theta),
\qquad
\gamma_s^\mathrm{tar}\in
\argmax_\gamma \tilde q(s,x,h_s^\mathrm{tar},\gamma;\theta),
\end{align}
then $H^\mathrm{tar}$ and $\Gamma^\mathrm{tar}$ are optimal and $\tilde v$ is the optimal value function.
\end{proposition}

\begin{proof}{Proof}
The result follows the martingale-characterization argument of Theorem 7 in \citet{jiaQLearningContinuousTime}, specialized to the present LQG game and deterministic Markov controls. 
\end{proof}
 
\begin{remark}
The second statement in Proposition
\ref{prop:qlearning:Martingale_condition} connects behavior and target policies.
From the definition of the $q$-function at \eqref{eq:q:func} and the Bellman--Isaacs PDE \eqref{eq:BellmanIsaacs:u}, which forms the foundations of Theorem~\ref{theo:main_previouspaper}, one can see that the candidate policies given at \eqref{eq:hhat}-\eqref{eq:gammahat:2} of Theorem~\ref{theo:main_previouspaper} coincide with those at \eqref{eq:h_gamma:argmin} in Proposition~\ref{prop:qlearning:Martingale_condition}.
\end{remark}

As a corollary to Proposition \ref{prop:qlearning:Martingale_condition}, we have:

\begin{corollary}[Martingale Orthogonality Condition]\label{coro:qlearning:M:mart:orthogonality}
Let $(\tilde{M}_s)_{s \in [t,T]}$ be the square-integrable $\mathbb{P}^\Gamma$-martingale defined at \eqref{eq:qlearning:Martingale:M}, for policies $H= \left(h_s\right)_{s \in [t,T]} \in \mathcal{A}^H_T, \Gamma = \left(\gamma_s\right)_{s \in [t,T]} \in \mathcal{A}^{\Gamma}_T$, and given a fixed or random initial state value $X_t=x$. Furthermore, let $\tilde q:[t,T]\times\mathbb R^n\times\mathbb R^m\times\mathbb R^d\to\mathbb R$ satisfying \eqref{eq:prop:qlearning:Martingale_condition:tildeq} without $\tilde q$ being identically zero. Then, for any test function $\xi \in L_\mathcal{F}^2 \left([t,T];\tilde{M};\mathbb{P}^\Gamma\right)$,
\begin{align}\label{eq:qlearning:MartOrthoCond}
    \mathbf{E}_{t,x}^{\mathbb{P}^\Gamma} \left[ \int_t^s \xi_u d\tilde{M}_u\right] 
    = \mathbf{E}_{t,x}^{\mathbb{P}^\Gamma} \left[ \int_t^s \xi_u \left\{ d\tilde{v} \left(u,X_u\right)
    + \left[\theta g(u,X_u,h_u, \gamma_u;\theta) - \tilde{q}(u,X_u,h_u, \gamma_u;\theta) \right]du\right\}\right] =0
                         \end{align}
\end{corollary}

\subsection{Parametrized Function Approximation for Policy Gradient Methods (Episodic Tasks)}\label{sec:PolicyGradientMethods:Intro}

In practice, the value functions $v$ and $q$ associated with the strategies $H$ and $\Gamma$ are not known. Policy gradient methods learn the parameters of a function approximation of the policies and value function. This section introduces the parametrization we will use to approximate the policies $H$ and $\Gamma$, the state-value function $v$, and the action-value function $q$.

We start with the function approximation of the policies. For our partly model-based setting, Theorem~\ref{theo:main_previouspaper} shows that the \emph{optimal policies} are affine in the state variable. To keep the discussion general, we consider a parametrization $H^\phi = \left(h(s,x;\phi)\right)_{s\in [t,T]} \in \mathcal{A}^H_T$ and $\Gamma^\phi = \left(\gamma(s,x;\phi)\right)_{s\in [t,T]}  \in \mathcal{A}^\Gamma_T$ of the policies, where the controls $h(s,x;\phi) \in \mathbb{R}^m$ and $\gamma(s,x;\phi) \in \mathbb{R}^d$ are affine in the state $x$, with time-varying coefficients parametrized by a vector $\phi \in \mathbb{R}^{d_\Phi}$: 
\begin{align}\label{eq:policy:param:general}
    h(s,x;\phi) :=& f^{(h,1)}\left(s;\phi^{(h,1)}\right) + f^{(h,2)}\left(s;\phi^{(h,2)}\right) x
                            \nonumber\\
    \gamma(s,x;\phi) :=& f^{(\gamma,1)}\left(s;\phi^{(\gamma,1)}\right) + f^{(\gamma,2)}\left(s;\phi^{(\gamma,2)}\right) x,
\end{align}    
where $f^{(h,1)} : [0,T] \times \mathbb{R}^{d_\phi^{(h,1)}} \to \mathbb{R}^m$, $f^{(h,2)} : [0,T] \times \mathbb{R}^{d_\phi^{(h,2)}} \to \mathbb{R}^{m\times n}$, $f^{(\gamma,1)} : [0,T] \times \mathbb{R}^{d_\phi^{(\gamma,1)}} \to \mathbb{R}^d$, $f^{(\gamma,2)} : [0,T] \times \mathbb{R}^{d_\phi^{(\gamma,2)}} \to \mathbb{R}^{d\times n}$ are $C^{1,1}$ functions. We restrict attention to parameter values $\phi$ for which the induced Markov controls belong to $\mathcal A_T^H$ and $\mathcal A_T^\Gamma$. The derivatives $\frac{\partial h}{\partial \phi}$ and $\frac{\partial \gamma}{\partial \phi}$ are matrices of dimension $m \times d_\phi$ and $d \times d_\phi$, respectively\footnote{\cite{silverDeterministicPolicyGradient2014} already noted that, in the case of deterministic policy, the derivative of the policy with respect to its parameters is a matrix. By contrast, in the case of stochastic policies, the derivative of the policy with respect to its parameters is typically a vector.}.

Theorem~\ref{theo:main_previouspaper} shows that the optimal value function $u$ is quadratic in the state variable with time-varying parameters, so we consider the following quadratic parametrization $v^\psi$ of the state-value function $v$, given parametrized policies $H^\phi \in \mathcal{A}^H_T$, $\Gamma^\phi \in \mathcal{A}^\Gamma_T$:
\begin{align}\label{eq:paramfunc:v}
    v^\psi(t,x)
    =& -\theta \left(\frac{1}{2}x'\tilde{Q}(t;\psi) x + \tilde{q}(t;\psi)' x + \tilde{k}(t;\psi)\right),
\end{align}
where $\tilde{Q}: [0, T] \times \mathbb{R}^{d_\psi} \to \mathbb{R}^{n\times n}$, $\tilde{q}: [0, T] \times \mathbb{R}^{d_\psi} \to \mathbb{R}^{n}$, and $\tilde{k}:  [0, T] \times \mathbb{R}^{d_\psi} \to  \mathbb{R}$ are $C^{1,1}$ functions parametrized by a $d_\psi$-element vector $\psi$, such as polynomial functions of time $t$ where the vector $\psi$ contains the polynomial's coefficients. For example, we could choose Taylor expansions of $Q$, $q$, and $k$ of degree $d_P$, requiring a total of $d_\psi = (d_P+1) \times (n^2 + n + 1)$ parameters.

\begin{remark}
The optimal value function $u$ and the optimal policies $H^*$ and $\Gamma^*$ depend on the risk-sensitivity parameter $\theta$ both explicitly and implicitly. This implicit dependence occurs via the coefficients $Q_s$, $q_s$, and $k_s$, defined at \eqref{eq:Q:Riccati}-\eqref{eq:k:integral} in Theorem~\ref{theo:main_previouspaper}, and whose equations depend on $\theta$. In the definition of the parametrized function $v^\psi$ at \eqref{eq:paramfunc:v}, the learned parameters $\tilde{Q}_s$, $\tilde{q}_s$, $\tilde{k}_s$ absorb the function $v$'s entire dependence on $\theta$.  
\end{remark}

The \emph{action-value} function $q$ can be parametrized by using its definition in \eqref{eq:q:func}. In fact, given our parametrization of the state-value function $v$ as a quadratic function of the state $x_t$ with time-varying coefficients generated by sufficiently smooth functions of time $(C^1[0,T])$, the
partial derivatives of $v^{\psi}$ are themselves at most quadratic in the state $x$ with coefficients that are (at least) continuous in $t$. Moreover,
the function $g$, which is defined at \eqref{eq:g}, is linear in $x$ and quadratic in $h$ and $\gamma$. Hence, we can consider a parametric approximation of $q$ in the following form
\begin{align}
    q^w\left( s,x, h, \gamma \right)
    :=& f^{(q,0)}(s;w)
    + f^{(q,1)}(s;w) x
    + x' f^{(q,2)}(s;w) x
    + \gamma' f^{(q,3)}(s;w)
    + \gamma' f^{(q,4)}(s;w) x
    + h' f^{(q,5)}(s;w)
                                    \nonumber\\
    &+ h' f^{(q,6)}(s;w) x
    + h' f^{(q,7)}(s;w) \gamma
    + h' f^{(q,8)}(s;w) h
    - \gamma' f^{(q,9)}(s;w) \gamma,
\end{align}
where $f^{(q,i)}$ are $C^{1,1}$ functions with appropriate dimensions, and $w \in \mathbb{R}^{d_w}$ is a vector of parameters that needs to be estimated.

\begin{remark}
The matrix-valued function $f^{(q,9)}$ need not be estimated if one imposes the known structure of the risk-sensitive game. Since $\theta g$ contains the term $-\frac{1}{2}\gamma'\gamma$, the direct quadratic contribution in $\gamma$ is obtained by setting $f^{(q,9)}=\frac{1}{2} I_d$. This restriction can therefore be hard-coded.
\end{remark}

\subsection{Actor--Critic via $q$-Learning for Episodic Tasks}\label{sec:PolicyGradientMethods:qlearning:Episodic}

The family of algorithms we discuss here parallels the semi-gradient {\it Temporal Difference (TD)} algorithms for discrete-time MDPs. As is typical of TD methods, this family of algorithms updates the time $t$ parameters based on reward information at time $s>t$ (respectively time $t_i$ and $t_{i+1}$ in discrete time). The update may be performed offline or online in a fully analogous way. Here, we focus on the online implementation (see \eqref{eq:qlearning:update:orth:online:CT} below).

Following \citet{jiaQLearningContinuousTime}, we use the martingale orthogonality condition in Corollary \ref{coro:qlearning:M:mart:orthogonality} to estimate the critic functions $v^\psi$ and $q^w$. The choice of the $\mathcal F_t$-adapted test functions specifies the resulting semi-gradient TD algorithm. For the critic, we use
\begin{align}\label{xizeta}
    \xi_t := \frac{\partial v^\psi(t,X_t)}{\partial \psi},
    \qquad
    \zeta_t := \frac{\partial q^w(t,X_t,h_t,\gamma_t)}{\partial w}.
\end{align}
These choices produce TD-style stochastic-approximation updates for the critic parameters $\psi$ and $w$.

The orthogonality condition then implies that
\begin{align}\label{eq:martcond:update}
     & \mathbf{E}_{t,x}^{\mathbb{P}^\Gamma} \left[ \int_t^s \xi_u d\tilde{M}_u\right] 
                            \nonumber\\
    =& \mathbf{E}_{t,x}^{\mathbb{P}^\Gamma} \left[ \int_t^s \xi_u \left\{ dv^{\psi} \left(u,X_u\right)
    + \left[\theta  g(u,X_u,h^\phi_u,\gamma^\phi_u;\theta) - q^w(u,X_u,h^\phi_u,\gamma^\phi_u;\theta)\right]du\right\}\right] 
    =  0,
\end{align}
for $\xi$ and analogously for $\zeta$.

The actor update is treated separately. Since the policies are deterministic maps from the time-state space into the action spaces, we use deterministic policy gradients, in the spirit of \citet{silverDeterministicPolicyGradient2014}. For the minimizing actor $h$ and maximizing actor $\gamma$, define
\begin{align}\label{eta_h_gamma}
    \eta_t^{(h)}
    &:=
    \left(
    \frac{\partial h^\phi(t,X_t)}
    {\partial \phi^{(h)}}
    \right)'
    \nabla_h q^w(t,X_t,h_t,\gamma_t),
    \\
    \eta_t^{(\gamma)}
    &:=
    \left(
    \frac{\partial \gamma^\phi(t,X_t)}
    {\partial \phi^{(\gamma)}}
    \right)'
    \nabla_\gamma q^w(t,X_t,h_t,\gamma_t).
\end{align}

Define the discrete martingale residual
\begin{align}\label{eq:martresidual:delta_k}
    \delta_k
    :=
    v^\psi(t_{k+1},x_{k+1})
    -
    v^\psi(t_k,x_k)
    +
    \theta g(t_k,x_k,h_k,\gamma_k;\theta)\Delta t
    -
    q^w(t_k,x_k,h_k,\gamma_k)\Delta t.
\end{align}
This martingale residual is used to evaluate the current policies and update the critic parameters. In implementation, $g_k = g(t_k,x_k,h_k,\gamma_k;\theta)$ is estimated from the observed data (state, securities prices, and benchmark level), actions, and value function estimates.

The critic is updated using the semi-gradient TD rules
\begin{align}
    \psi
    \leftarrow
    \psi
    +
    \alpha_\psi
    \xi_k
    \delta_k,
    \qquad
    w
    \leftarrow
    w
    +
    \alpha_w
    \zeta_k
    \delta_k
    \Delta t.
    \label{eq:qlearning:update:orth:online:CT},
\end{align}
with $\xi_k$ and $\zeta_k$ defined as in \eqref{xizeta} and evaluated at $(t_k,x_k,h_k,\gamma_k)$. The factor $\Delta t$ in the update for $w$ may equivalently be absorbed into the critic learning rate $\alpha_w$. 

The actor is then updated by a deterministic policy gradient step:
\begin{align}
    \phi^{(h)}
    &\leftarrow
    \phi^{(h)}
    -
    \alpha_h
    \eta_k^{(h)},
    \qquad
    \eta_k^{(h)}
    :=
    \left(
    \frac{\partial h^\phi(t_k,x_k)}
    {\partial \phi^{(h)}}
    \right)'
    \nabla_h q^w(t_k,x_k,h_k,\gamma_k),
    \label{eq:qlearning:update:orth:online:CT:phi:h}\\
    \phi^{(\gamma)}
    &\leftarrow
    \phi^{(\gamma)}
    +
    \alpha_\gamma
    \eta_k^{(\gamma)},
    \qquad
    \eta_k^{(\gamma)}
    :=
    \left(
    \frac{\partial \gamma^\phi(t_k,x_k)}
    {\partial \phi^{(\gamma)}}
    \right)'
    \nabla_\gamma q^w(t_k,x_k,h_k,\gamma_k).
    \label{eq:qlearning:update:orth:online:CT:phi:gamma}
\end{align}
with $h$ and $\gamma$ evaluated at the current time-state pair $(t,x) \in \left[0,T\right] \times \mathbb{R}^n$, and where  $\alpha_h, \alpha_\gamma,\alpha_\psi, \alpha_w >0 $ are learning rates. Equation \eqref{eq:qlearning:update:orth:online:CT:phi:h} contains a minus because the update descends toward the minimizing control $h$, while equation \eqref{eq:qlearning:update:orth:online:CT:phi:gamma} contains a plus because the update ascends toward the maximizing control $\gamma$. Note that the actor update separates the policy gradient into two multiplicative components: the Jacobian of the action with respect to the parameter vector $\Phi$ and the local action gradients $\nabla_h q^w$ and $\nabla_\gamma q^w$. The latter, which represent the gradients of $q$ with respect to the action $h$ and $\gamma$, respectively, are computed once \(q^w\) has been updated. This separation mirrors the usual actor--critic architecture: the critic is trained by a TD residual, while the actor is improved using the critic's action-gradient.

For implementation, partition $[0,T]$ into $K$ intervals of length $\Delta t = \frac{T}{K}$, with grid $\mathcal I={t_0,\ldots,t_{K-1}}$. The strategies $H=(h_s){s\in\mathcal I}$ and $\Gamma=(\gamma_s){s\in\mathcal I}$ are then updated on this grid. Algorithm~\ref{alg:actorcritic:qlearning:online} presents the online on-policy $q$-learning actor–critic method. In the algorithm, the continuous-time increments $dv$ and $dt$ are replaced by $v^\psi(t_{k+1},x_{k+1})-v^\psi(t_k,x_k)$ and $\Delta t$, and the learning decay function $\ell$ slows the updates across episodes. Section~\ref{sec:Implementation} presents an off-policy proof of concept using behavior policies to optimize target policies.

\begin{algorithm}
\caption{Episodic Online, On-Policy Actor--Critic via $q$-Learning}\label{alg:actorcritic:qlearning:online}
\begin{algorithmic}
    \Require
        \State Parametrized functions $f^{h,1}$, $f^{h,2}$, $f^{\gamma,1}$, $f^{\gamma,2}$, $\tilde{Q}$, $\tilde{q}$, $\tilde{k}$, and $f^{q,0}, \ldots, f^{q,9}$
        \State Input a time sequence $0 = t_0 < t_1 < \dots < t_k < \ldots < t_K = T$ indexed by $k \in \mathbb{N}, k = 1,\ldots,K$
        \State \texttt{episode} $\gets$ number of episodes
        \State Input learning rates $\alpha_h$, $\alpha_\gamma$, $\alpha_\psi$, $\alpha_w$, risk sensitivity parameter $\theta$, and learning decay function $\ell\,(\texttt{episode})$
        \State Initialize the vectors $\phi = \begin{pmatrix} \phi^{(h)} \\ \phi^{(\gamma)} \end{pmatrix}$, $\psi$, $w$, and the time 0 value of the state process $x_0$
        
    \For{$\texttt{episode} = 1,\ldots,N$}
        \For{$k=0,\ldots,K-1$}
            \State $\ell_{\texttt{episode}} \gets \ell(\texttt{episode})$ 
            \State \textit{// Generate actions and get data:}
            \State Generate actions $h_k=h^\phi(t_k,x_k)$ and $\gamma_k=\gamma^\phi(t_k,x_k)$ 

            \State \textit{// Update the environment and compute the performance}
            
            \State Estimate $g_k=g(t_k,x_k,h_k,\gamma_k;\theta)$ from observed data (state, securities, and benchmark), actions, and value function estimates.

            \State Compute critic gradients:
            \(
            \xi_k=\partial_\psi v^\psi(t_k,x_k),
            \qquad
            \zeta_k=\partial_w q^w(t_k,x_k,h_k,\gamma_k).
            \)
            \State Compute actor gradients:
            \begin{align*}
            \eta_k^{(h)}
            =
            \left(\partial_{\phi^{(h)}}h^\phi(t_k,x_k)\right)'
            \nabla_h q^w(t_k,x_k,h_k,\gamma_k),
            \quad
            \eta_k^{(\gamma)}
            =
            \left(\partial_{\phi^{(\gamma)}}\gamma^\phi(t_k,x_k)\right)'
            \nabla_\gamma q^w(t_k,x_k,h_k,\gamma_k).
            \end{align*}
            \State Observe $x_{k+1}$.
            \State Compute:
            $\delta_k
            =
            v^\psi(t_{k+1},x_{k+1})
            -
            v^\psi(t_k,x_k)
            +
            \theta g_k\Delta t
            -
            q^w(t_k,x_k,h_k,\gamma_k)\Delta t.$

            \State Update critics:
            $\psi\gets\psi+\ell_{\texttt{episode}} \alpha_\psi\xi_k\delta_k,
            \quad
            w\gets w+\ell_{\texttt{episode}} \alpha_w\zeta_k\delta_k\Delta t.$
            \State Update actors:
            $\phi^{(h)}
            \gets
            \phi^{(h)}
            -
            \ell_{\texttt{episode}} \alpha_h\eta_k^{(h)},
            \quad
            \phi^{(\gamma)}
            \gets
            \phi^{(\gamma)}
            +
            \ell_{\texttt{episode}} \alpha_\gamma\eta_k^{(\gamma)}.$
        \EndFor
    \EndFor
\end{algorithmic}
\end{algorithm}

\subsection{Ergodic Problem and Actor--Critic via $q$-Learning for Continuing Tasks}\label{sec:PolicyGradientMethods:qlearning:Continuing}

According to \eqref{eq:SDG:ergodic}, the long-run average criterion $\mathcal{J}_{\infty}$ associated with the strategies $H= \left(h_t\right)_{t \geq 0} \in \bar{\mathcal{A}}^H$, $\Gamma = \left(\gamma_t\right)_{t \geq 0} \in \bar{\mathcal{A}}^{\Gamma}$ is
\begin{align}\label{eq:valuefunc:v:ergodic}
    \mathcal{J}_{\infty}(x,H,\Gamma;\theta)   
    := \limsup_{T \to \infty} T^{-1}\mathbf{E}_{x}^{\mathbb{P}^\Gamma} \left[ \theta \int_{0}^{T} g(X_t,h_t,\gamma_t;\theta) \, dt \right],
\end{align}
where $g$ is defined at \eqref{eq:g}. The corresponding game value is
\begin{align}\label{eq:ergodic:game:value:RL}
    \rho
    =
    \inf_{H\in\bar{\mathcal A}^H}
    \sup_{\Gamma\in\bar{\mathcal A}^{\Gamma}}
    \mathcal J_\infty(x,H,\Gamma;\theta).
\end{align}
In line with \eqref{eq:BellmanIsaacs:ergodic}--\eqref{eq:rho}, the value function $v$, together with the scalar long-run value $\rho$, solves the ergodic PDE
\begin{align}\label{eq:PDEcharacterization:v:ergodic}
    \left[b + B x + \Lambda \gamma \right]' \frac{\partial v}{\partial x}
    + \frac{1}{2} \tr \left(\Lambda\Lambda' \frac{\partial^2 v}{\partial x^2} \right)
    +\theta g(x,h,\gamma;\theta) 
    = \rho, \; x \in \mathbb{R}^n.
\end{align}
From~\eqref{eq:baru}, we get a quadratic approximation of the value function $v$,
that, in line with \eqref{eq:paramfunc:v}, we denote by $\tilde{v}(x) = -\theta \left( \frac{1}{2} x' \tilde{Q} x + \tilde{q}'x \right)$.
We also get the following approximation for the ergodic $q$-function:
\begin{align}\label{eq:q:func:ergodic}
    q\left(x,h,\gamma;\theta\right)
    := 
    \left[b + B x + \Lambda \gamma \right]' \frac{\partial v}{\partial x} 
    + \frac{1}{2} \tr \left(\Lambda\Lambda' \frac{\partial^2 v}{\partial x^2} \right)
    +\theta g(x,h,\gamma;\theta)
    - \rho.
\end{align}
An ergodic analogue of Proposition~\ref{prop:qlearning:Martingale_condition} follows from the same martingale-characterization argument; see also \citet[Theorem~9]{jiaQLearningContinuousTime}. The main changes are the admissible classes and the additional long-run scalar value \(\rho\).

Algorithm \ref{alg:actorcritic:qlearning:online} for episodic tasks extends with minor modifications to continuing tasks. The three main differences are: the form of the value functions differs, as discussed above; the algorithm for continuing tasks only loops through time; and the time dependence disappears from the parametrization. This last element leads to several simplifications. The policy parametrization is of the form 
\begin{align}
        h(x;\phi) 
        = \phi^{(h,1)} + \phi^{(h,2)} x,
         \qquad
        \gamma(x;\phi) 
        = \phi^{(\gamma,1)} + \phi^{(\gamma,2)} x 
\end{align}
where $\phi^{(h,1)} \in \mathbb{R}^m, \phi^{(h,2)} \in \mathbb{R}^{m \times n}$, $\phi^{(\gamma,1)} \in \mathbb{R}^d, \phi^{(\gamma,2)} \in \mathbb{R}^{d \times n}$ are constants. Now, the parameterization for the state value and action value functions becomes
$
    v^\psi(x)
    = -\theta \left(\frac{1}{2}x'\Psi x + \psi' x \right)
$,
for $\Psi \in \mathbb{R}^{n \times n}, \psi \in \mathbb{R}^n$, and 
\begin{align}
    q^w\left(x, h, \gamma \right)
    :=& w_0
    + w_x' x
    + x' W_{xx} x
    + w_\gamma'\gamma
    + \gamma' W_{\gamma x} x
    + w_h' h
    + h' W_{hx} x
    + h' W_{h\gamma} \gamma
    + h' W_{hh} h
    - \gamma' W_{\gamma\gamma} \gamma,
\end{align}
where all coefficients have conformable dimensions, with \(W_{xx}\), \(W_{hh}\), and \(W_{\gamma\gamma}\) taken symmetric when the corresponding quadratic terms are identified.

\subsection{The Kelly Criterion and Its Implication for Explainability}\label{sec:PolicyGradientMethods:Explainability}

The actor–critic framework also clarifies how the learned allocation can be explained. The Kelly criterion is not only a limiting case of the risk-sensitive benchmarked problem; it provides a natural reference allocation. While reinforcement-learning algorithms produce an allocation for a set of learned parameters, they do not by themselves explain \emph{why} that allocation is recommended. This is especially true for long-term asset allocation algorithms. The structure of the benchmarked risk-sensitive solution provides such an explanation. 

Assuming we know the value of the model's coefficients, \citet{LleoRunggaldier_RSIMviaDuality_2026} establishes that the optimal asset allocation has the form of a \emph{fractional Kelly strategy}, consisting of a combination of three funds: the Kelly portfolio with allocation $h^{K}(s, X_s)$, a benchmark-tracking portfolio with allocation $h^\mathrm{Bench}(s, X_s)$, and an Intertemporal Hedging Portfolio (IHP) with allocation $h^\mathrm{IHP}(s, X_s)$. Mathematically,
\begin{align}
h^*(s,X_s)
=
f\,h^K(s,X_s)
+
(1-f)\,h_s^{\mathrm{Bench}}
-
(1-f)\,h^{\mathrm{IHP}}(s,X_s),
\label{eq:summary_fractional_kelly}
\end{align}
where $f:=\frac1{\theta+1}$, $h^{K}(s, X_s)$ is given at \eqref{eq:hhat:Kelly}, and
\begin{align}
h_s^{\mathrm{Bench}}
=
\left(\Sigma_s\Sigma_s'\right)^{-1}
\Sigma_s\Xi_s,
\qquad
h^{\mathrm{IHP}}(s,X_s)
=
\left(\Sigma_s\Sigma_s'\right)^{-1}
\Sigma_s\Lambda_s'
\left(q_s+Q_sX_s\right),
\end{align}
where $Q_s$ and $q_s$ are respectively the Riccati matrix and vector coefficients appearing in the quadratic value function. Importantly, the relative allocation between the three funds is fixed and depends solely on the investor's risk-sensitivity, a known quantity.     

Therefore, if we either know or can learn the asset allocation within these three constituent funds, we can decompose the asset allocation $h_s^{RL}$ produced by a reinforcement learning algorithm,  such as the actor--critic algorithm presented above, as:
\begin{align}
    h_s^{RL} = \frac{1}{\theta+1} h^{K}(s, X_s) + \frac{\theta}{\theta+1} h^\mathrm{Bench}(s, X_s) -\frac{\theta}{\theta+1} h^\mathrm{IHP}(s, X_s) + h^{\epsilon}, 
\end{align}
where $h^{\epsilon}$ is a residual term capturing approximation, estimation, and numerical errors.

The Kelly portfolio can be learned independently using an actor--critic approach, as discussed above. Three simplifications occur: (1) the problem is not a game, so there is no adversarial control $\gamma$, (2) we do not need to change measure, and (3) the state is uncontrolled.

The allocation to the \emph{benchmark-tracking portfolio} can be interpreted as the minimum-variance hedge ratio produced by hedging the benchmark $L$ using the $m$ securities $S_i, i=1,\ldots,m$ from the asset market. Therefore, $h^\mathrm{Bench}$ can be estimated online as the slope vector of a sequential multiple linear regression of the benchmark returns against the securities returns\footnote{Alternatively, one could estimate the quadratic variations $\Sigma_s\Sigma_s'$ and cross variation $\Sigma_s\Xi_s$ sequentially, and then multiply these two quantities to obtain an estimate for $h_s^\mathrm{Bench}$.}.  

We can use a similar approach for the \emph{IHP}. Decompose the asset allocation of the IHP, $h_s^\mathrm{IHP}$ into two multiplicative terms, $\left(\Sigma_s\Sigma_s'\right)^{-1}\Sigma_s\Lambda_s'$ and $\left(q_s+Q_s X_s\right)$. The first term can be interpreted as the minimum-variance hedge ratio obtained by hedging the state variables $X$ using the $m$ securities $ S_i$, $ i=1,\ldots,m$. So we can learn it using a sequential multiple linear regression of the state variables against the securities returns. The second term involves the learned estimates of the Riccati matrix and vector in the parametrized value function, which are already produced by the actor--critic algorithm.

Therefore, for $k$ clients and $\ell\leq k$ benchmarks, the explanatory decomposition requires $k+1$ actor--critic runs, one per client and one for the Kelly portfolio, together with $\ell+1$ sequential regressions, one per benchmark and one for the common IHP hedge. The additional cost is therefore limited relative to learning the client-specific policies.

\section{Implementation}
\label{sec:Implementation}

This section presents a proof of concept for the actor--critic reinforcement-learning framework developed in Section~\ref{sec:PolicyGradientMethods}, focusing on a continuing task associated with an ergodic investment problem. The aim is not to propose a fully model-free implementation, but to show that the analytical model and its actor--critic formulation lead to a concrete computational procedure. We use financial market data to compute the parameters for the analytical ergodic risk-sensitive benchmarked asset management model derived in Section \ref{sec:RSBAM:infinite} that will serve as a reference to assess the learning speed and accuracy of the actor--critic algorithm. Specifically, this section tests whether a proof-of-concept algorithm with a reduced critic, that is, a critic reduced to the set of parameters required by the actors, can learn the analytical model's optimal policies from the ground up, using only data simulated from the analytical model. The implementation is carried out in R.

\subsection{Investment Problem and Financial Data}

We consider a U.S. equity fund with risk-sensitivity parameter $\theta=1$. The benchmark allocates 90\% to the S\&P 500 and 10\% to the S\&P 400. The manager uses a six-factor model consisting of the five Fama--French factors \citep{FF2015} and Momentum \citep{Carhart1997}.

The investment universe contains $m=13$ U.S. exchange-traded funds. The first 11 ETFs track the S\&P 500 GICS sectors, so the S\&P 500 can be replicated by holding these ETFs at their sector weights. The remaining two ETFs are the iShares Core S\&P 400 Mid-Cap ETF (IJH) and the iShares Core S\&P 600 Small-Cap ETF (IJR). Since the S\&P 600 is not included in the benchmark, any allocation to small-capitalization stocks is tactical.

Model parameters are estimated from daily returns over June 20, 2018, to December 31, 2024, giving $\mathcal T=1644$ observations and $\Delta t=1/252$. June 20, 2018, is the earliest date for which all series are available after the GICS sector revision. The sample spans a volatile period in U.S. equity markets, including the 2019 rise, the COVID-19 crash in March 2020, and the subsequent rally. Factor and money-market returns are obtained from Kenneth French's data library\footnote{The data are available online at \url{http://mba.tuck.dartmouth.edu/pages/faculty/ken.french/Data_Library/f-f_factors.html}.}, while ETF returns are taken from the CRSP database\footnote{\copyright 2026 Center for Research in Security Prices (CRSP), LLC.}.

\subsection{Proof-of-Concept Reinforcement-Learning Algorithm}
\label{subsec:FullDiagnostic4B}

The proof-of-concept algorithm focuses on learning the two actors, namely the investment policy $H$ and the adversarial policy $\Gamma$, rather than on approximating the full scalar action-value function. This choice reflects both the analytical structure of the saddle point and the economic role of the controls. The portfolio allocation $h$ is the economically relevant decision variable, while $\gamma$ is the auxiliary adversarial control induced by the duality representation. In this implementation, the ergodic relative value function is fixed at its analytical form,
$\bar u(x)
    =
    -\theta
    \left(
    \frac{1}{2} x' \bar Q x + \bar q'x
    \right)
$, where $\bar Q$ and $\bar q$ are expressed in standardized coordinates. The critic learns only the actor-relevant gradients of the action-value function.

The portfolio and adversarial actors are affine in the augmented standardized state
\begin{align}
    x_{\mathrm{aug}}=(1,x')',
    \qquad
    h^\phi(x)=\Phi^{(h)}x_{\mathrm{aug}},
    \qquad
    \gamma^\phi(x)=\Phi^{(\gamma)}x_{\mathrm{aug}}.
\end{align}
The corresponding analytical saddle-point actors are denoted by $\Phi^{(h),*}$ and $\Phi^{(\gamma),*}$. Exploration is introduced through behavior controls
\begin{align}
    h^\mathrm{beh}(x)=h^\phi(x)+\epsilon^h,
    \qquad
    \gamma^\mathrm{beh}(x)=\gamma^\phi(x)+\epsilon^\gamma,
\end{align}
where $\epsilon^h$ and $\epsilon^\gamma$ are mean-zero perturbations.

Since the analytical $q$-function is quadratic in $(x,h,\gamma)$, its gradients with respect to $h$ and $\gamma$ are affine in $z = (1,x',h',\gamma')'$. We therefore use the reduced critic
\begin{align}\label{eq:reducedcritic:gradient}
    \widehat{\nabla_h q}(x,h,\gamma)=M_hz,
    \qquad
    \widehat{\nabla_\gamma q}(x,h,\gamma)=M_\gamma z,
\end{align}
where $M_h\in\mathbb R^{m\times(1+n+m+d)},
\;
M_\gamma\in\mathbb R^{d\times(1+n+m+d)}$.

The matrices $M_h$ and $M_\gamma$ are learned by recursive least squares from behavior-action observations. To this end, we construct the TD scalar target from the fixed relative value function:
\begin{align}
    q_k^{TD}
    = 
    \frac{\bar u(x_{k+1})-\bar u(x_k)}{\Delta t}
    +
    \theta g(x_k,h_k^\mathrm{beh},\gamma_k^\mathrm{beh};\theta),
\end{align}
where, n this proof-of-concept implementation, $x_{k+1}$ is generated from the standardized dynamics
\begin{align}
    x_{k+1}
    =
    x_k
    + \left(b + B x_k + \Lambda \gamma_k^\mathrm{beh} \right) \Delta t
    + \Lambda \varepsilon_k \sqrt{\Delta t},
    \qquad
    \varepsilon_k\sim N(0,I_d).
\end{align}
and $g(x_k,h_k^\mathrm{beh},\gamma_k^\mathrm{beh};\theta)$ is computed using the analytical formula \eqref{eq:g}.

Central finite differences of $q_k^{TD}$, evaluated at the behavior controls, generate targets $\widehat{\nabla_h q}^{\,TD}_k, \; \widehat{\nabla_\gamma q}^{\,TD}_k$, which are used with the behavior feature vector
$
z_k^\mathrm{beh}
=
\left(1,x_k',(h_k^\mathrm{beh})',(\gamma_k^\mathrm{beh})'\right)'
$ to update $M_h$ and $M_\gamma$. 

The updated critic is then evaluated at the target controls $(h_k,\gamma_k)$ to update the actors:
\begin{align}
\Phi^{(h)}
\leftarrow
    \Phi^{(h)}
    -
    \alpha_h
    (\theta\Sigma\Sigma')^{-1}
    \widehat{\nabla_h q}_k
    \frac{x_{\mathrm{aug},k}'}{1+\|x_{\mathrm{aug},k}\|^2},
\end{align}
\begin{align}
    \Phi^{(\gamma)}
    \leftarrow
    \Phi^{(\gamma)}
    +
    \alpha_\gamma
    \widehat{\nabla_\gamma q}_k
    \frac{x_{\mathrm{aug},k}'}{1+\|x_{\mathrm{aug},k}\|^2}.
\end{align}
The signs reflect the minimization over $h$ and maximization over $\gamma$. The preconditioner $(\theta\Sigma\Sigma')^{-1}$ is consistent with the local curvature of the analytical Hamiltonian in the portfolio direction, and the factor $(1+\|x_{\mathrm{aug},k}\|^2)^{-1}$ stabilizes the affine actor update.

The proof-of-concept implementation learns 2,993 entries: 224 for the two actors, 1,248 for the two gradient critics, and 1,521 for the auxiliary RLS covariance matrix $P$. The value-function coefficients $\bar Q$ and $\bar q$ are fixed at their analytical values. Algorithm~\ref{alg:Proof_of_concept_algo} presents the step-by-step procedure. 

\begin{algorithm}
\caption{Proof-of-concept online off-policy algorithm for the continuing case}
\label{alg:Proof_of_concept_algo}
\begin{algorithmic}[1]
\Require Standardized state sample \(\{x_i\}\); analytical coefficients \(\bar Q,\bar q\); model coefficients \(a,A,b,B,\Sigma,\Lambda,\Xi,c,C\); risk-sensitivity \(\theta\); learning rates \(\alpha_h,\alpha_\gamma\); time step \(\Delta t\).
\State Initialize \(M_h,M_\gamma\) and the RLS covariance matrix \(P\).
\State Initialize actors \(\Phi^{(h)}\) and \(\Phi^{(\gamma)}\).
\For{\(k=1,\ldots\)}
    \State Sample a standardized state \(x_k\) and set \(x_{\mathrm{aug},k}=(1,x_k')'\).
    \State Compute current actor actions:
    \(
    h_k=\Phi^{(h)}x_{\mathrm{aug},k},
    \quad
    \gamma_k=\Phi^{(\gamma)}x_{\mathrm{aug},k}.
    \)
    \State Generate exploratory behavior actions from the current actions:
    \(
    h_k^\mathrm{beh}=h_k+\epsilon_k^{h},
    \quad
    \gamma_k^\mathrm{beh}=\gamma_k+\epsilon_k^{\gamma},
    \)
    with mean-zero exploratory perturbations $\epsilon_k^{h},\epsilon_k^{\gamma}$.
    \State Generate a one-step transition
    \(
    x_{k+1}
    =
    x_k+
    \left(b + B x_k + \Lambda \gamma_k^\mathrm{beh}\right)\Delta t
    +
    \Lambda \varepsilon_k\sqrt{\Delta t}.
    \)
    
    \State Form the TD scalar target at the behavior actions:
    \(
    q^{TD}_k
    =
    \frac{\bar u(x_{k+1})-\bar u(x_k)}{\Delta t}
    +
    \theta g(x_k,h_k^\mathrm{beh},\gamma_k^\mathrm{beh};\theta).
    \)
    \State Compute central finite-difference approximations from the TD scalar target:
    $
    \widehat{\nabla_h q}^{\,TD}_k,
    \;
    \widehat{\nabla_\gamma q}^{\,TD}_k.
    $
    \State Define the behavior feature vector at the behavior actions:
    $
    z_k^\mathrm{beh}=(1,x_k',(h_k^\mathrm{beh})',(\gamma_k^\mathrm{beh})')'.
    $
    \State Update \(M_h\) and \(M_\gamma\) by RLS using the targets
    $\widehat{\nabla_h q}^{\,TD}_k$ and
    $\widehat{\nabla_\gamma q}^{\,TD}_k$ and critic feature vector $z_k^\mathrm{beh}$.
    \State Evaluate learned critic gradients at the target actor actions:
    \[
    \widehat{\nabla_h q}_k=M_h(1,x_k',h_k',\gamma_k')',
    \qquad
    \widehat{\nabla_\gamma q}_k=M_\gamma(1,x_k',h_k',\gamma_k')'.
    \]
    \State Update the portfolio actor:
    \(
    \Phi^{(h)}
    \leftarrow
    \Phi^{(h)}
    -
    \alpha_h
    (\theta\Sigma\Sigma')^{-1}
    \widehat{\nabla_h q}_k
    \frac{x_{\mathrm{aug},k}'}{1+\|x_{\mathrm{aug},k}\|^2}.
    \)
    \State Update the adversarial actor:
    \(
    \Phi^{(\gamma)}
    \leftarrow
    \Phi^{(\gamma)}
    +
    \alpha_\gamma
    \widehat{\nabla_\gamma q}_k
    \frac{x_{\mathrm{aug},k}'}{1+\|x_{\mathrm{aug},k}\|^2}.
    \)
\EndFor
\end{algorithmic}
\end{algorithm}

\subsection{Numerical Outcome}

The following results are obtained from the R implementation of Algorithm~\ref{alg:Proof_of_concept_algo}. Figure~\ref{fig:error_convergence} shows the convergence of the estimation errors for the critic and actor matrices, and for the corresponding actions. Table~\ref{tab:diagnostic_errors} reports the final relative errors.

\begin{figure}[!ht]
\begin{center}
\includegraphics[width=0.70\textwidth]{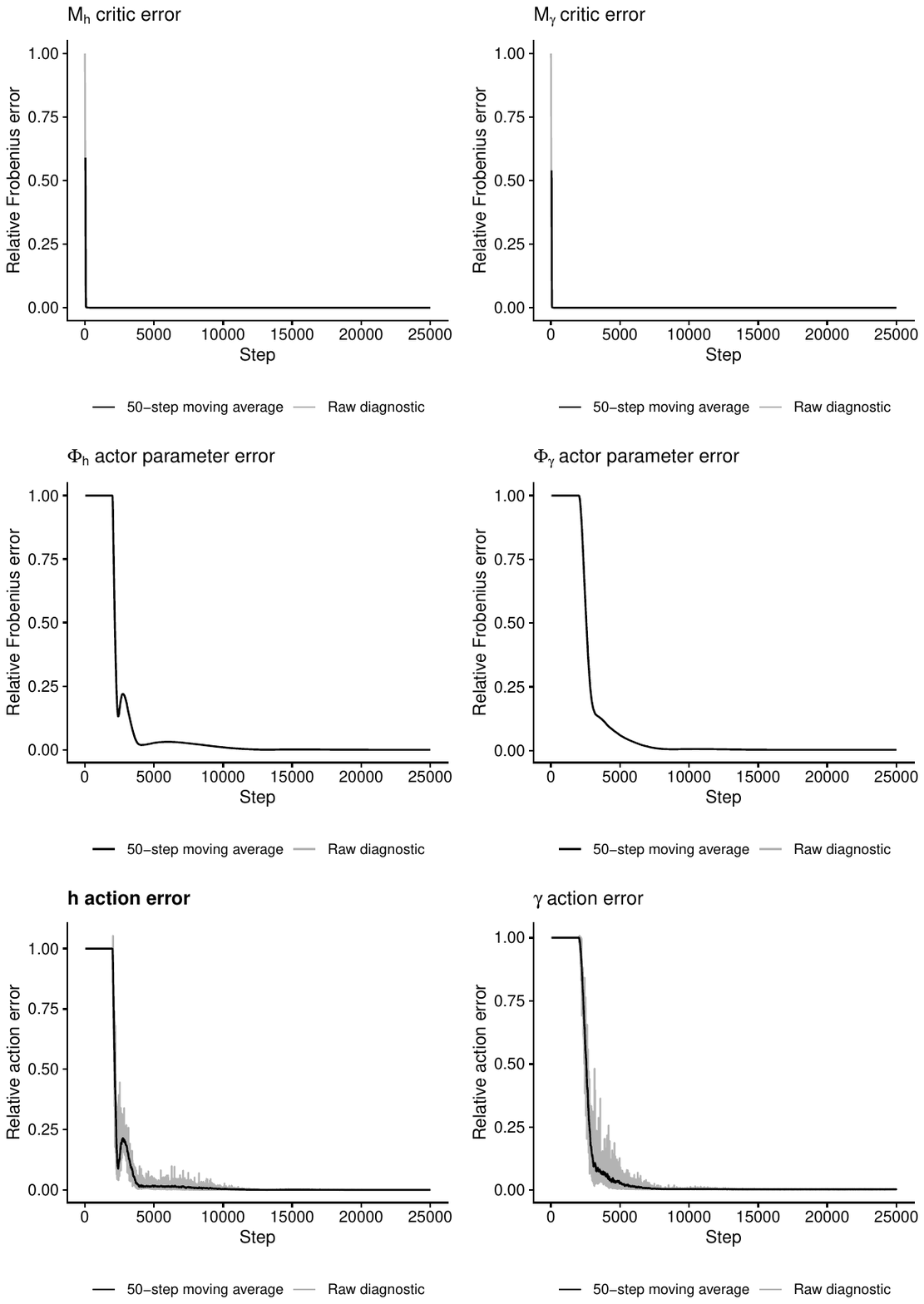}
\caption{Error diagnostics for the critic and actors}
\label{fig:error_convergence}
\end{center}
\end{figure}

\begin{table}[htbp]
\centering
\begin{tabular}{llll}
\toprule
Diagnostic & Formula & Portfolio control $h$ & Adversarial control $\gamma$ \\
\midrule
Critic matrix error
& $\frac{\|M_\cdot-M_\cdot^*\|_F}{\|M_\cdot^*\|_F}$
& $1.47\times10^{-6}$
& $2.57\times10^{-4}$ \\

TD gradient-target error
& $\left\|
    \widehat{\nabla_\cdot q}^{\,TD}
    -
    \nabla_\cdot q
    \right\|$
& $7.02\times10^{-13}$
& $1.07\times10^{-3}$ \\

Actor matrix error
& $\frac{\|\Phi^{(\cdot)}-\Phi^{(\cdot),*}\|_F}{\|\Phi^{(\cdot),*}\|_F}$
& $7.22\times10^{-5}$
& $3.35\times10^{-3}$ \\

Action error
& $\frac{\|\cdot^\phi(x)-\cdot^*(x)\|}{\|\cdot^*(x)\|}$
& $5.69\times10^{-5}$
& $3.28\times10^{-3}$\\
\bottomrule
\end{tabular}
\caption{Final diagnostic errors in the proof-of-concept implementation. In the formulas of the `Formula' column, the dot `$\cdot$' is a placeholder for the control and takes the value $h$ for the portfolio control or $\gamma$ for the adversarial control. The matrices $M_h^*$ and $M_\gamma^*$ denote the analytical gradient-critic matrices induced by the optimal $q$-function in standardized coordinates. For a matrix $M$, $\|M\|_F=\sqrt{\tr(M'M)}$ denotes the Frobenius norm.}
\label{tab:diagnostic_errors}
\end{table}

The portfolio-gradient critic is learned essentially exactly, whereas the adversarial-gradient critic retains a small residual error. This asymmetry is structural. Under the $\mathbb P^\Gamma$-dynamics, the standardized state $x_t$ is governed by the adversarial control $\gamma$, whereas the portfolio allocation $h$ does not enter the state equation. Therefore, the TD-generated gradient with respect to $h$ differentiates only the running reward term and is not affected by transition-gradient effects. By contrast, $\gamma$ enters the transition dynamics. A one-step TD finite-difference approximation involves $\Lambda'D\bar u(x_{t+\Delta t})$ rather than the infinitesimal generator term $\Lambda'D\bar u(x)$. Since
$ 
    D\bar u(x) = -\theta(\bar Qx+\bar q)    
$, the finite-step discrepancy is approximately
$
    -\theta\Lambda'\bar Q(x_{t+\Delta t}-x)
$. This explains why the $\gamma$-actor carries a small residual error even when the $h$-actor is essentially exact.

From an asset-management perspective, this asymmetry is favorable. The economically relevant decision variable is the portfolio allocation $h$, while $\gamma$ is the auxiliary adversarial control introduced by the free energy--entropy duality. The diagnostics indicate that the duality-based formulation places the transition-learning burden primarily on $\gamma$, while leaving the portfolio actor with a clean and accurate TD learning signal.

Overall, these diagnostics provide a proof of concept for the partly model-based actor--critic framework. The reduced gradient critic exploits the quadratic structure of the action-value function, while the affine actors exploit the analytical form of the saddle-point controls. In the data-calibrated ETF allocation problem considered here, the resulting algorithm efficiently learns the optimal policy under the simulated standardized dynamics, thereby supporting the computational relevance of the theoretical actor--critic formulation.

\section{Conclusion}\label{sec:Conclusion}

This paper develops a reinforcement-learning approach to a risk-sensitive benchmarked asset allocation problem in a partly model-based setting. Using free energy--entropy duality, we reformulate the nonstandard benchmarked control problem as a linear--quadratic--Gaussian stochastic differential game under an equivalent probability measure, yielding explicit finite- and infinite-horizon solutions. The actor--critic formulation then uses this analytical structure to design interpretable actor and critic parametrizations while relaxing the assumption that all model coefficients are known. The resulting investment strategy admits a clear economic interpretation through fractional Kelly decompositions. A proof-of-concept implementation calibrated to financial market data illustrates the practical relevance of the approach. It also reveals a favorable structural asymmetry: the portfolio actor receives a clean TD learning signal because the portfolio allocation does not enter the state transition, whereas the auxiliary adversarial actor carries the residual error associated with the finite-step transition approximation.

\appendix

\section{Proof of Theorem~\ref{theo:ergodic}}\label{app:proof:theo:ergodic}

\paragraph{Admissibility of the candidate strategies.}

We start by establishing that the candidate strategies $H^*$ and $\Gamma^*$ are Borel-measurable and Markov. For the value function $\bar u(x)$ in Theorem~\ref{theo:ergodic}, the candidate strategies $H^*$ and $\Gamma^*$ induced by \eqref{eq:hbar} and \eqref{eq:gammabar:2} become
\begin{align}
    h^*(x)
    =&  
    \frac{1}{\theta+1}\left(\Sigma\Sigma'\right)^{-1}\left[
        \left(A - \theta \Sigma\Lambda' \bar{Q} \right)x
        + a
        + \theta \Sigma (\Xi - \Lambda' \bar{q})
        \right]
                                \label{eq:hhat:verif}\\
    \gamma^*(x) 
    =&  
    \left\{
        - \frac{\theta}{\theta+1} \Sigma'\left(\Sigma\Sigma'\right)^{-1}A 
        - \theta \mathcal{P}^{-}\Lambda' \bar{Q}
    \right\}x 
    - \frac{\theta}{\theta+1} \Sigma'\left(\Sigma\Sigma'\right)^{-1}a
     + \theta \mathcal{P}^{-} (\Xi - \Lambda' \bar{q}).
                            \label{eq:gammahat:verif}
\end{align}  

Therefore, the defining policies $h^*(\cdot)$ and $\gamma^*(\cdot)$ are Borel-measurable, and the induced strategies $H^*$ and $\Gamma^*$ are Markov. Moreover, under $\mathbb{P}^{\Gamma^*}$, the state process $X = \left(X_t\right)_{t \geq 0}$ remains Gaussian. To see this, substitute \eqref{eq:gammahat:verif} into \eqref{eq:state:Pgamma:FO} to obtain the following dynamics:
\begin{align}\label{eq:state:Pgamma:ergodic}
    d X_t
    =& \left\{ 
        \left( \kappa - K_0 \bar{q} \right)
        + K_2 X_t \right\}dt
        + \Lambda dW^{\Gamma^*}_t,
\end{align}
where $K_0$ is given at \eqref{eq:notation:K0} and we defined $K_2 := K_1 -K_0 \bar Q$ and $\kappa 
:= b 
- \frac{\theta}{\theta+1} \Lambda\Sigma'\left(\Sigma\Sigma'\right)^{-1}a
+ \theta \Lambda \mathcal{P}^{-}\Xi$ for compactness.    

Next, we consider the mean $m_t$ and covariance $P_t$ of $X_t$ under $\mathbb{P}^{\Gamma^*}$. They solve
\begin{align}\label{eq:ODE:mt_and_Pt}
    \frac{dm_t}{dt} 
    = \left( \kappa - K_0 \bar{q} \right) + K_2 m_t.
    \qquad
    \frac{dP_t}{dt} 
    =  K_2 P_t + P_t K_2' + \Lambda\Lambda',
\end{align}
For deterministic $X_0=x$,
$    
    m_t 
    = e^{K_2 t}x  
    + \int_0^t e^{ K_2(t-s)}\left( \kappa - K_0 \bar{q}\right) \, ds
$ and
$
    P_t 
    = \int_0^t 
        e^{K_2 (t-s)} \Lambda\Lambda' e^{ K_2'(t-s)} \, ds
$. If $X_0$ has covariance $P_0$, the initial condition $e^{K_2 t}P_0e^{K_2't}$ must be added to the solution for $P_t$.

By Theorem 4.2 in \citet{kuna02}, the matrix $K_1 - K_0 Q_T$ converges, as $T\to\infty$, to the stable matrix $K_2=K_1-K_0 \bar Q$. Consequently, the limits of $m_T$ and $P_T$ exist as $T\to\infty$. Moreover, if \eqref{eq:ctrlcond:BA:theorem} holds, then $(K_2,\Lambda)$ is controllable. If \eqref{eq:theo:ergodic:mainassumption} also holds, Lemma 5.1 in \citet{kuna02} gives $\theta\bar Q^{-1}>\bar P\geq P_T$. Thus, the limiting covariance matrix $\bar P$ exists and is nondegenerate.

Equations \eqref{eq:hbar} and \eqref{eq:state:Pgamma:ergodic} show that $h^*(x)$ and $\gamma^*(x)$ ar affine in $x$. Since $X$ is Gaussian under $\mathbb P^{\Gamma^*}$, $h^*(X_s)$ is square-integrable on every interval $[0,T]$, and $H^*\big|_{[0,T]}\in\mathcal A_T^H$ for all $T>0$.

It remains to check the Girsanov condition for $\Gamma^*$. For convenience, define the density process associated with $\Gamma^*$ on $[0,T]$ by
\begin{align}\label{eq:chi:GammaStar:ts}
\chi_t^{\Gamma^*}
:=
\exp\left\{
-\frac12\int_0^t \|\gamma^*(X_s)\|^2\,ds
+\int_0^t \gamma^*(X_s)' \, dW_s
\right\},
\qquad t\in[0,T].
\end{align}
Since $\gamma^*(x)$ is affine and the state equation is linear, the stochastic exponential is a true martingale on finite horizons by the standard linear-growth martingale criterion for Girsanov densities. Hence $\chi^{\Gamma^*}$ is an exponential martingale on every $[0,T]$. Conditions (i)-(iii) in Definition \ref{def:classAgammaT:fullobs} follow, so $\Gamma^* \big|_{[0,T]} \in \mathcal{A}^{\Gamma}_T$ for all $T>0$, so $\Gamma^* \in \bar{\mathcal{A}}^{\Gamma}$.

\paragraph{Verification of the saddle-point inequalities.}

We next prove the saddle-point inequalities. Let
\begin{align}
    F(x,h,\gamma)
    :=
    \mathcal L^{h,\gamma}\bar u(x) +\theta g(x,h,\gamma;\theta).
\end{align}
Since $(\bar h,\bar\gamma)$ is a pointwise saddle point of the Hamiltonian, the ergodic Bellman--Isaacs equation implies, for all admissible $h$ and $\gamma$,
\begin{align}
    F(x,\bar h(x),\gamma)
    \leq
    \rho
    \leq
    F(x,h,\bar\gamma(x)).
\end{align}

Let $\Gamma\in\bar{\mathcal A}^{\Gamma}$ be arbitrary and fix the minimizing policy $H^*$. By It\^o's formula under $\mathbb P^\Gamma$,
\begin{align}
    d\bar u(X_t)
    =
    \mathcal L^{h^*,\gamma}\bar u(X_t) \, dt
    +D\bar u(X_t)'\Lambda \, dW_t^\Gamma.
\end{align}
Taking expectations and using the martingale property of the stochastic integral gives
\begin{align}
    \mathbf E_x^{\mathbb P^\Gamma}
    \left[
    \int_0^T
    \mathcal L^{h^*,\gamma}\bar u(X_t) \, dt
    \right]
    =
    \mathbf E_x^{\mathbb P^\Gamma}[\bar u(X_T)] - \bar u(x).
\end{align}
Since
$\mathcal L^{h^*,\gamma}\bar u(X_t) + \theta g(X_t,h_t^*,\gamma_t;\theta) \leq \rho$,
we obtain
\begin{align}
    \frac{1}{T}
    \mathbf E_x^{\mathbb P^\Gamma}\left[
        \theta\int_0^T g(X_t,h_t^*,\gamma_t;\theta) \,dt
    \right]
    \leq
    \rho 
    + \frac{\bar u(x)}{T} 
    - \frac{1}{T} \mathbf E_x^{\mathbb P^\Gamma}[\bar u(X_T)].
\end{align}
By the definition of $\bar{\mathcal A}^{\Gamma}$, the first two moments of $X_T$ under $\mathbb P^\Gamma$ remain bounded, and in fact have finite limits as $T\to\infty$. Since $\bar u$ is quadratic, it follows that
$
    \lim_{T\to\infty}
    \frac{1}{T} \mathbf E_x^{\mathbb P^\Gamma}[\bar u(X_T)]
    = 0
$. Therefore,
\begin{align}
\limsup_{T\to\infty}
    \frac1T
    \mathbf E_x^{\mathbb P^\Gamma}
    \left[
    \theta\int_0^T
    g(X_t,h_t^*,\gamma_t;\theta) \, dt
    \right]
    \leq
    \rho.
\end{align}

Conversely, let $H\in\bar{\mathcal A}^H$ be arbitrary and fix the maximizing policy $\Gamma^*$. Since
\begin{align}
    \mathcal L^{h,\gamma^*}\bar u(X_t) + \theta g(X_t,h_t,\gamma_t^*;\theta)
    \geq \rho,
\end{align}
the same It\^o argument under $\mathbb P^{\Gamma^*}$ gives
\begin{align}
    \frac{1}{T} \mathbf E_x^{\mathbb P^{\Gamma^*}}\left[
        \theta\int_0^T g(X_t,h_t,\gamma_t^*;\theta) \, dt    
    \right]
    \geq
    \rho 
    + \frac{\bar u(x)}{T}
    - \frac{1}{T} \mathbf E_x^{\mathbb P^{\Gamma^*}}[\bar u(X_T)].
\end{align}
The admissibility of $\Gamma^*$ implies that the final term again vanishes after division by $T$. Hence
\begin{align}
    \liminf_{T\to\infty} \frac{1}{T} \mathbf E_x^{\mathbb P^{\Gamma^*}} \left[
        \theta\int_0^T g(X_t,h_t,\gamma_t^*;\theta) \, dt
    \right]
    \geq
    \rho,
\end{align}
and hence the corresponding $\limsup$ inequality in Theorem~\ref{theo:ergodic} also holds.

Taking $H=H^*$ and $\Gamma=\Gamma^*$, the pointwise saddle-point equality and the same It\^o argument give
\begin{align}
    \lim_{T\to\infty}
    \frac{1}{T}
    \mathbf E_x^{\mathbb P^{\Gamma^*}}
    \left[
    \theta\int_0^T g(X_t,h_t^*,\gamma_t^*;\theta)\,dt
    \right]
    =
    \rho.
\end{align}

\paragraph{Optimality.}
The saddle-point inequalities show that $H^*$ attains the infimum and $\Gamma^*$ attains the supremum in the ergodic game. Hence the pair $(H^*,\Gamma^*)$ is optimal.

%
%


\end{document}